\pdfoutput=1  
\documentclass[]{article}  
\usepackage{url,float}
\usepackage{graphicx}
\usepackage{amsmath}
\usepackage{amsfonts}
\usepackage{amssymb}
\usepackage{latexsym}

\newcommand{\hide}[1]{}

\newcommand{\ABox}{
\raisebox{3pt}{\framebox[6pt]{\rule{6pt}{0pt}}}
}
\newenvironment{proof}{{\bf Proof:}}{\hfill\ABox}

\newtheorem{theorem}{{\bf Theorem}}

\newtheorem{lemma}{Lemma}

\newcommand{\lemlab}[1]{\label{lemma:#1}}
\newcommand{\thmlab}[1]{\label{thm:#1}}

\newcommand{\tablab}[1]{\label{tab:#1}}
\newcommand{\figlab}[1]{\label{fig:#1}}
\newcommand{\seclab}[1]{\label{sec:#1}}

\newcommand{\lemref}[1]{\ref{lemma:#1}}
\newcommand{\thmref}[1]{\ref{thm:#1}}

\newcommand{\secref}[1]{\ref{sec:#1}}
\newcommand{\figref}[1]{\ref{fig:#1}}
\newcommand{\tabref}[1]{\ref{tab:#1}}

{\makeatletter
 \gdef\xxxmark{%
   \expandafter\ifx\csname @mpargs\endcsname\relax 
     \expandafter\ifx\csname @captype\endcsname\relax 
       \marginpar{xxx}
     \else
       xxx 
     \fi
   \else
     xxx 
   \fi}
 \gdef\xxx{\@ifnextchar[\xxx@lab\xxx@nolab}
 \long\gdef\xxx@lab[#1]#2{{\bf [\xxxmark #2 ---{\sc #1}]}}
 \long\gdef\xxx@nolab#1{{\bf [\xxxmark #1]}}
 \gdef\turnoffxxx{\long\gdef\xxx@lab[##1]##2{}\long\gdef\xxx@nolab##1{}}%
}

\def\P{{\mathcal P}}

\def\g{{\gamma}}

\def\L{{\Lambda}}

\def\o{{\omega}}
\def\O{{\Omega}}

\def\e{{\varepsilon}}
\def\a{{\alpha}}
\def\b{{\beta}}

\def\bP{{\partial P}}

\def\R{{\mathbb{R}}}

\newcommand{\squeezelist}{\setlength{\itemsep}{0pt}}

\title{%
Conical Existence of Closed Curves \\
on Convex Polyhedra
} 

\author{%
Joseph O'Rourke%
    \thanks{Department of Computer Science, Smith College, Northampton, MA
      01063, USA.
      \protect\url{orourke@cs.smith.edu}.}
\and
Costin V\^{i}lcu%
    \thanks{Institute of Mathematics ``Simion Stoilow'' 
      of the Romanian Academy,
      P.O. Box 1-764,
      RO-014700 Bucharest, Romania.
    \protect\url{Costin.Vilcu@imar.ro}.
}
}

\begin{document}
\maketitle

\begin{abstract}
Let $C$ be a simple, closed, directed curve on the surface of
a convex polyhedron $\P$.
We identify several classes of curves $C$ that ``live on a cone,''
in the sense that $C$ and a neighborhood to one side may be
isometrically
embedded on the surface of a cone $\L$, with the apex $a$ of $\L$
enclosed inside (the image of) $C$;
we also prove that each point of $C$ is ``visible to'' $a$.
In particular, we obtain that these curves have non-self-intersecting
developments in the plane.
Moreover, the curves we identify that live on cones to both sides
support
a new type of ``source unfolding'' of the entire surface of $\P$ to
one
non-overlapping piece, as reported in a companion paper.
\end{abstract}

\section{Introduction}
\seclab{Introduction}
Let $\P$ be the surface of a convex polyhedron,
and let
$C$ be any simple, closed, directed curve on $\P$.
In this paper we address the question of which curves $C$ 
``live on a cone'' to either or both sides.
We first explain this notion, which is based on neighborhoods of $C$.

\paragraph{Living on a Cone.}
An open region $N_L$ is a \emph{vertex-free neighborhood} of $C$ 
to its left if its
right boundary is $C$, and it contains no vertices of $\P$.
In general $C$ will have many vertex-free left neighborhoods, and all
will
be equivalent for our purposes.
We say that $C$ \emph{lives on a cone} to its left
if there exists a cone $\L$ 
and a neighborhood $N_L$
so that $C \cup N_L$ may be embedded isometrically
onto $\L$, and
encloses the cone apex $a$.

A \emph{cone} is a developable surface with curvature zero everywhere
except at one point, its \emph{apex}, which has total incident surface
angle, called the \emph{cone angle}, of at most $2 \pi$.
Throughout, we will consider a cylinder as a cone whose apex is at
infinity with cone angle 0,
and a plane as a cone with apex angle $2 \pi$.
We only care about the intrinsic properties of the cone's surface;
its shape in $\R^3$ is not relevant for our purposes.  So one could
view it as having a circular cross section, although we will often
flatten
it to the plane, in which case it forms a doubly covered triangle
with apex angle half the cone angle.
Except in special cases, the cone $\L$ is unrelated to any cone that
may be formed by extending the faces of $\P$ to the left of $C$.

\begin{figure}[htbp]
\centering
\includegraphics[width=0.5\linewidth]{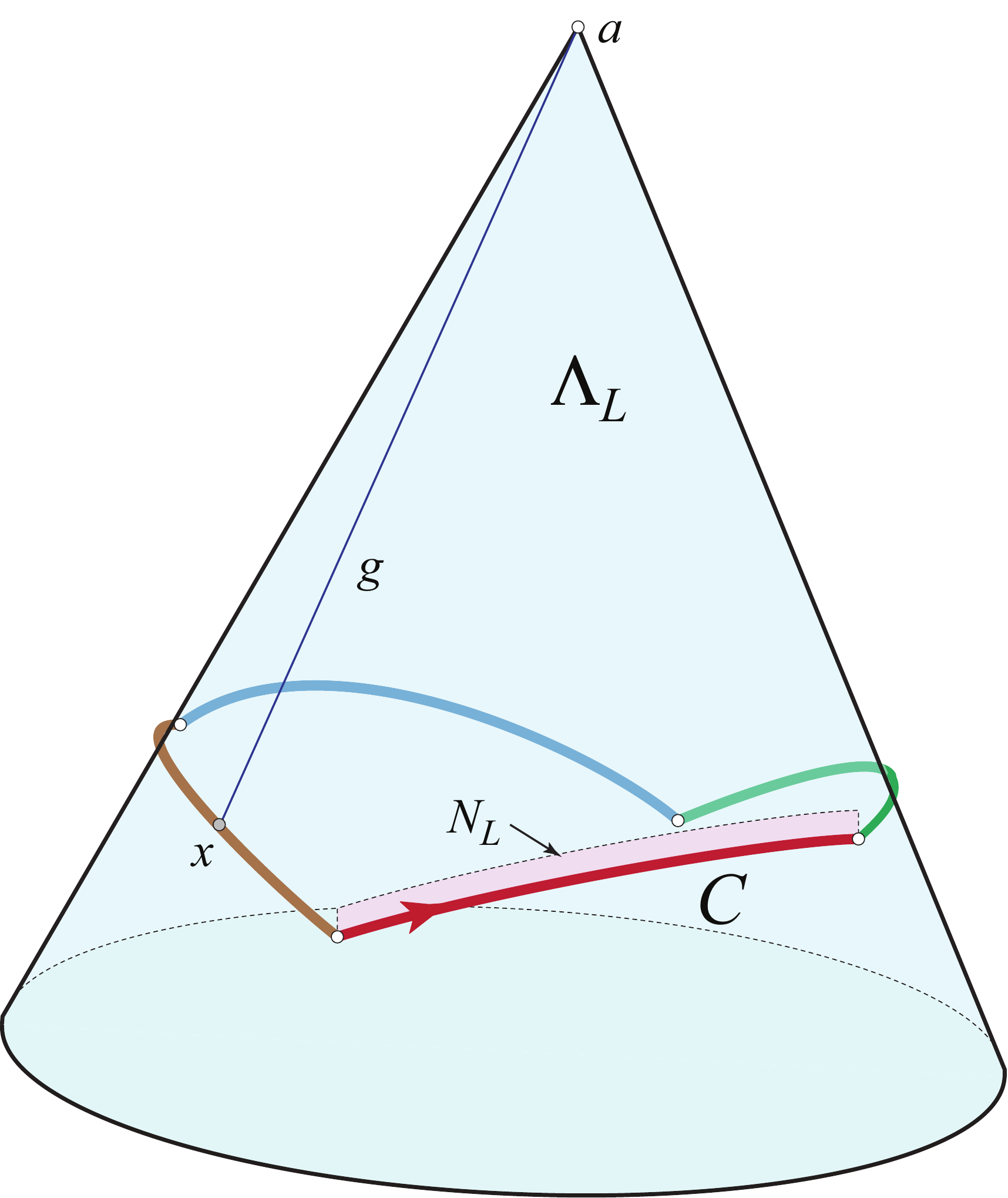}
\caption{A 4-segment curve $C$ which lives
on cone $\L_L$ to its left.
A portion of $N_L$ is shown, 
and a generator $g=ax$ is illustrated.
}
\figlab{Cone3D}
\end{figure}
To say that $C \cup N_L$ embeds isometrically
into $\L$ means that we could cut out $C \cup N_L$ and paste it onto
$\L$
with no wrinkles or tears: the distance between any two points of
$C \cup N_L$ on $\P$ is the same as it is on $\L$.
See Figure~\figref{Cone3D}.
We say that $C$ lives on a cone to its right if $C \cup N_R$ embeds on
the
cone, where $N_R$ is a right neighborhood of $C$
such that the cone apex $a$ is inside (the image of) $C$.
We will call the cones $\L_L$ and $\L_R$ to the left and right of $C$
when we need to distinguish them.
We will see that all four combinatorial possibilities occur:
$C$ may not live on a cone to either side, it may live on a cone to
one
side but not to the other, it may live on different cones to its two
sides,
or live on the same cone to both sides.

\paragraph{Motivations.}
We have two motivations to study curves that live on a cone,
aside from their intrinsic interest.
First,
every simple, closed curve $C$ on a cone $\L$ may be 
\emph{developed} on the plane by rolling $\L$ and transferring
the ``imprint'' of $C$ to the plane.
This will allow us to strengthen a previous result on simple
(i.e., non-self-intersecting) developments of certain curves.
Second, for curves $C$ that live on a cone to both sides,
our results support a generalization of the ``source unfolding''
of a polyhedron.
Both of these motivations will be detailed further (with references)
in Section~\secref{Applications}. 

\paragraph{Curve Classes.}
To describe our results, we introduce a number of different classes of
curves on convex polyhedra,
which
exhibit different behavior with respect to living on a cone.
Altogether, we define eight classes of curves.
All our curves $C$ are simple (non-intersecting), closed, directed curves on a convex polyhedron
$\P$,
and henceforth we will generally drop these qualifications.

For any point
$p \in C$, 
let $L(p)$ be the total surface angle incident to $p$ at the left side of $C$,
and $R(p)$ the angle to the right side.
$C$ is a \emph{geodesic} if
$L(p)  {=} R(p)  {=} \pi$ for every point $p$ on $C$.
Generally this is called a \emph{closed geodesic} in the literature.
When a geodesic is extended on a surface and later crosses itself,
each closed portion generally forms what is known as a
\emph{geodesic loop}: $L(p) {=} R(p)  {=} \pi$
for all but
one exceptional \emph{loop point} $x$, at which it may be that
$L(x) {\neq} \pi$ or $R(x) {\neq} \pi$.
(The loop versions of curves are important because they are in general
easier to find than ``pure'' versions.)

Define a curve $C$ to be \emph{convex} (to the left) if
the angle to the left is
at most $\pi$ at every point $p$: $L(p) {\le} \pi$;
and say that $C$ is a \emph{convex loop} if this condition holds 
for all but
one exceptional \emph{loop point} $p$, at which
$L(p) {>} \pi$ is allowed.

A curve $C$ is a \emph{quasigeodesic} if it is
convex to both sides: $L(p) {\le} \pi$ and $R(p) {\le} \pi$
for all $p$ on $C$.
(This is a notion introduced by Alexandrov to allow geodesic-like
curves to pass through vertices of $\P$.)
A \emph{quasigeodesic loop} satisfies the same condition
except at an exceptional loop point $p$, at
which $L(p) {\le} \pi$ but $R(p) {>} \pi$ (or vice versa) is allowed.
Thus a quasigeodesic loop is convex to one side
and a convex loop to the other side.

Finally, define $C$ to be a \emph{reflex curve}%
\footnote{
    We opt for the term ``reflex'' rather than ``concave''
    for its greater syntactic difference from ``convex.''
}
if the angle to one side (we consistently use the right side) is
at least $\pi$  at every point $p$: $R(p) {\ge} \pi$;
and say that $C$ is a \emph{reflex loop} if this condition holds 
for all but
an exceptional loop point $p$, at which
$R(p) {<} \pi$.

The eight curve classes are then the four listed in the table below,
and their loop variations, which permit violation of the angle
conditions
at one point:
\begin{table}[htbp]
\begin{center}
\begin{tabular}{| c | c  |}
        \hline
\emph{Curve class} & \emph{Angle condition}
         \\ \hline \hline
geodesic & $L(p) = \pi = R(p)$
        \\ \hline
quasigeodesic & $L(p) \le \pi$ and $R(p) \le \pi$
        \\ \hline
convex & $L(p) \le \pi$ 
        \\ \hline
reflex & $R(p) \ge \pi$
        \\ \hline
\end{tabular}
\caption{Curve classes.}
\tablab{Curve.classes}
\end{center}
\end{table}
We now describe relations between the classes.
Most are obvious, following from the definitions.
All the non-loop curves are special cases of their loop version:
a geodesic is a geodesic loop, etc.
A geodesic is a quasigeodesic, and a quasigeodesic is convex to both
sides.
A geodesic loop is a quasigeodesic loop, which is convex to one side
and a convex loop to the other side.
To explain the relationship between convex and reflex curves,
we recall
the notion of ``discrete curvature,'' or simply ``curvature.''

The \emph{curvature} $\o(p)$ at any point $p \in \P$
is the ``angle deficit'':  
$2 \pi$ minus the sum of the face angles incident to $p$. 
The curvature is only nonzero at vertices of $\P$; 
at each vertex it is positive because $\P$ is convex. 
The curvature at the apex of a cone is similarly $2 \pi$ minus the
cone angle.

Define a \emph{corner} of curve $C$ to be any point $p$
at which either $L(p) {\neq} \pi$ or $R(p) {\neq} \pi$.
Let $c_1,c_2,\ldots,c_m$ be the corners of $C$,
which may or may not also be vertices of $\P$.
$C$ ``turns'' at each $c_i$, and is straight at any noncorner point.
Let $\a_i=L(c_i)$ be the surface angle to the left side at $c_i$,
and $\b_i=R(c_i)$ the angle to the right side.
Also let $\o_i = \o(c_i)$ to simplify notation.
We have $\a_i + \b_i + \o_i = 2 \pi$
by the definition of curvature.

Returning to our discussion of curve classes,
a convex curve that passes through no vertices of $\P$ is
a reflex curve to the other side, because $\o_i {=} 0$
and so $\a_i {\le} \pi$ implies that $\b_i {\ge} \pi$.
A convex curve that passes through at most one vertex of $\P$,
say at $c_m$,
is a reflex loop to the other side, with possibly $\b_m < \pi$,
and is a reflex curve to that side if $\a_m + \o_m \le \pi$
because then $\b_m \ge \pi$.
The relationship between convex and reflex is symmetric:
so a reflex curve that passes through no vertices is convex to the
other side, and a reflex curve that passes through one vertex is
a convex loop to the other side.
The other side of a reflex loop is a convex loop,
as will be discussed further in Section~\secref{Reflex} (cf.~Table~\tabref{Reflex}).

We illustrate some of these concepts in
Figure~\figref{IcosaCuboctaCurves}:
(a)~shows an icosahedron, and (b)~a cubeoctahedron.
For both polyhedra, $\o(v)=\frac{1}{3}\pi$ for each vertex $v$ of $\P$.
The curve illustrated in~(a) is convex to both sides, with
$\frac{2}{3}\pi$ to one side and $\pi$ to the other at each
of its five corners.  Thus it is a quasigeodesic.
The curve in~(b) is convex to one side, with
angles 
$$
( 
\tfrac{5}{6} \pi,  \tfrac{5}{6} \pi, \tfrac{1}{2}\pi,  
\tfrac{5}{6} \pi,  \tfrac{5}{6} \pi, \tfrac{1}{2}\pi 
)
$$
at its six corners,
but because the angles to the other side are (respectively)
$$
( 
\tfrac{5}{6} \pi,  \tfrac{5}{6} \pi, \tfrac{7}{6}\pi,  
\tfrac{5}{6} \pi,  \tfrac{5}{6} \pi, \tfrac{7}{6}\pi 
)
$$
it falls outside our classification system to that side
(because it violates convexity at two corners, and
reflexivity at four corners).
\begin{figure}[htbp]
\centering
\includegraphics[width=0.75\linewidth]{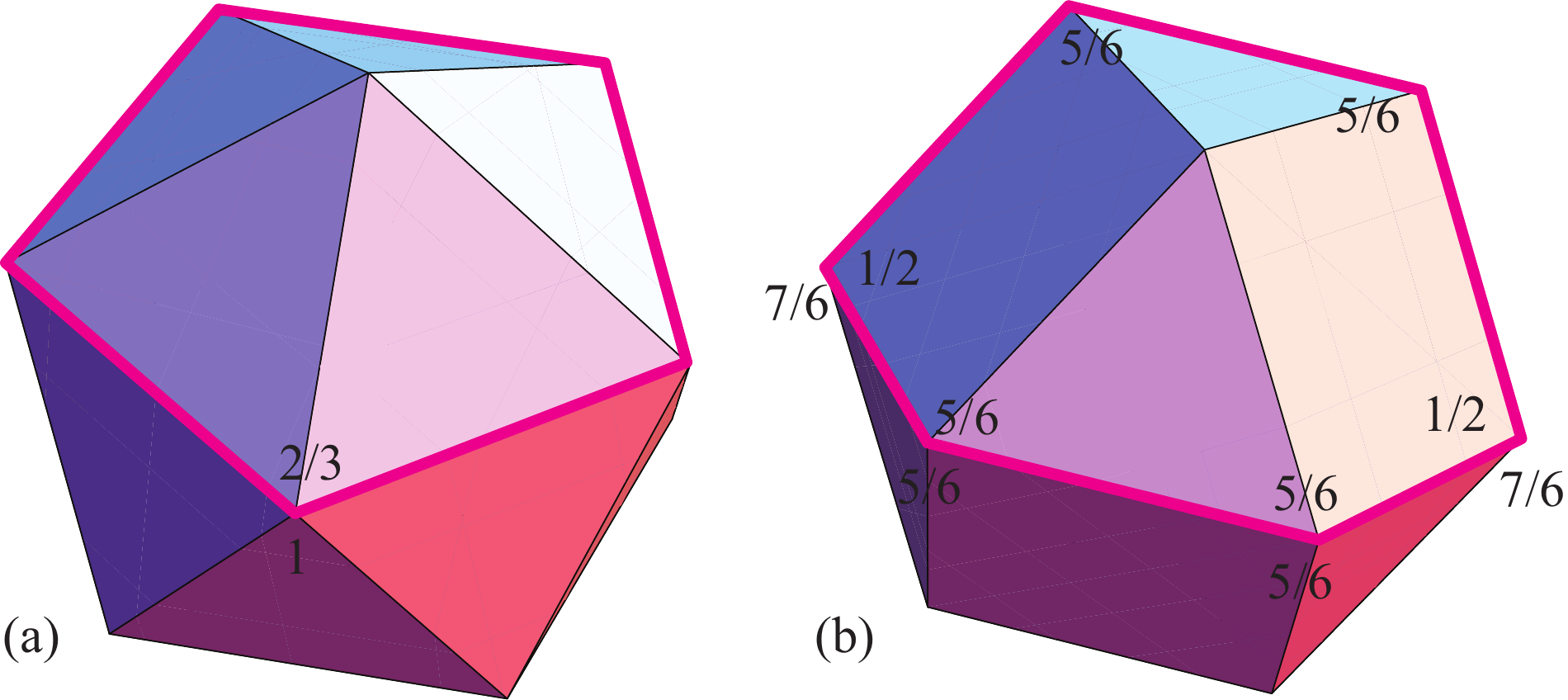}
\caption{(a)~Quasigeodesic curve on a Icosahedron.
(b)~Convex curve on a Cubeoctahedron. 
Angles are shown at vertices in units of $\pi$.
}
\figlab{IcosaCuboctaCurves}
\end{figure}

The main result of this paper is that a convex curve lives on a cone
to its convex side, and a reflex loop 
whose other side is convex
lives on a cone to its reflex
side.
One consequence is that any convex curve (which could be
a quasigeodesic) that includes at most one
vertex lives on a cone to both sides.
We also show that a convex loop might not live on a cone to its convex side.

\paragraph{Visibility.}
An additional property is needed for these cones to support
our applications.
A \emph{generator} of a cone $\L$ is a half-line 
starting from the apex $a$ and lying on $\L$.
A curve $C$ that lives on $\L$ is 
\emph{visible} from the apex if every generator 
meets $C$ at one point.%
\footnote{
   In other terminology, $C$ could be said to be
   \emph{star-shaped} from $a$.
}
See again Figure~\figref{Cone3D};
Figure~\figref{ConeDevelopment}(a) ahead illustrates a $C$ not
visible from $a$.
Although it is quite possible for a curve to live on a cone but
not be visible from its apex,
we establish that, for the classes we identify, $C$ is indeed
visible from the apex of the cone on which it lives.

\section{Preliminary Tools}
\seclab{PreliminaryTools}

\paragraph{The Gauss-Bonnet Theorem.}
We will employ this theorem in two forms.
The first is that  the total curvature of $\P$ is $4 \pi$:
the sum of $\o(v)$ for all vertices $v$ of $\P$ is $4 \pi$.
It will be useful to partition the curvature into three pieces.
Let $\O_L(C)=\O_L$ be the total curvature
strictly interior to the region of $\P$ to the left of $C$,
$\O_R$ the curvature to the right, and
$\O_C$ the sum of the curvatures on $C$ (which is
nonzero only at vertices of $\P$).
Then  $\O_L + \O_C + \O_R = 4\pi$.

The second form of the Gauss-Bonnet theorem relies on the 
notion of the ``turn'' of a curve.
Define $\tau_L(c_i) = \tau_i = \pi - \a_i$ as the left \emph{turn} of
curve $C$ at corner $c_i$,
and let $\tau_L(C) = \tau_L$ be the total (left) turn of $C$,
i.e., the sum of $\tau_i$ over all corners of $C$.
(The turn at noncorner points of $C$ is zero.  Note that the curve turn at a
point is not directly related to the surface curvature at that point.)
Thus a convex curve has nonnegative turn at each corner,
and a reflex curve has nonpositive turn at each corner.
Then $\tau_L + \O_L = 2 \pi$,
and defining the analogous term to the right of $C$,
$\tau_R + \O_R = 2 \pi$.
So, if $C$ is a geodesic, $\tau_L=\tau_R = 0$ and $\O_L = \O_R = 2 \pi$.

\paragraph{Alexandrov's Gluing Theorem.}
In our proofs we use 
Alexandrov's
celebrated theorem~\cite[Thm.~1, p.~100]{a-cp-05}
that gluing polygons to form a topological sphere in such a way that at
most $2\pi$ angle is glued at any point, results in a unique
convex polyhedron.

\paragraph{Vertex Merging.}
We now explain a technique used by
Alexandrov, e.g.,~\cite[p.~240]{a-cp-05}.
Consider two vertices $v_1$ and $v_2$ of curvatures $\o_1$ and $\o_2$ 
on $\P$,
with $\o_1+\o_2 < 2 \pi$,
and cut $\P$ along a shortest path $\g(v_1,v_2)$ joining $v_1$ to $v_2$.
Construct a planar triangle $T = \bar v' \bar v_1 \bar v_2$ such that
its base $\bar v_1 \bar v_2$ has the same length as $\g(v_1,v_2)$,
and the base angles are equal to 
$\frac{1}{2}\o_1$ and respectively $\frac{1}{2}\o_2$.
Glue two copies of $T$ along the corresponding lateral sides,
and further glue the two bases of the copies to the two ``banks'' of
the cut of $\P$ along $\g(v_1,v_2)$.
By Alexandrov's Gluing Theorem,
the result is a convex polyhedral surface $\P'$.
On $\P'$, the points $v_1$ and $v_2$ are no longer vertices
because exactly the angle deficit at each has been sutured in; 
they have been replaced by
a new vertex $v'$ of curvature $\o'=\o_1+\o_2$
(preserving the total curvature).
Figure~\figref{VertexMerging}(a) illustrates this.
Here $\g(v_1,v_2) = v_1 v_2$ is the top ``roof line'' of the
house-shaped polyhedron $\P$. Because
$\o_1 = \o_2 = \frac{1}{2}\pi$, $T$ has base angles
$\frac{1}{4}\pi$ and apex angle $\frac{1}{2}\pi$.
Thus the curvature $\o'$ at $v'$ is $\pi$.
(Other aspects of this figure will be discussed later.)

Note this vertex-merging procedure only works when $\o_1+\o_2 < 2 \pi$;
otherwise the angle at the apex $\bar v'$ of $T$ would be greater than
or equal to $\pi$.
\begin{figure}[htbp]
\centering
\includegraphics[width=\linewidth]{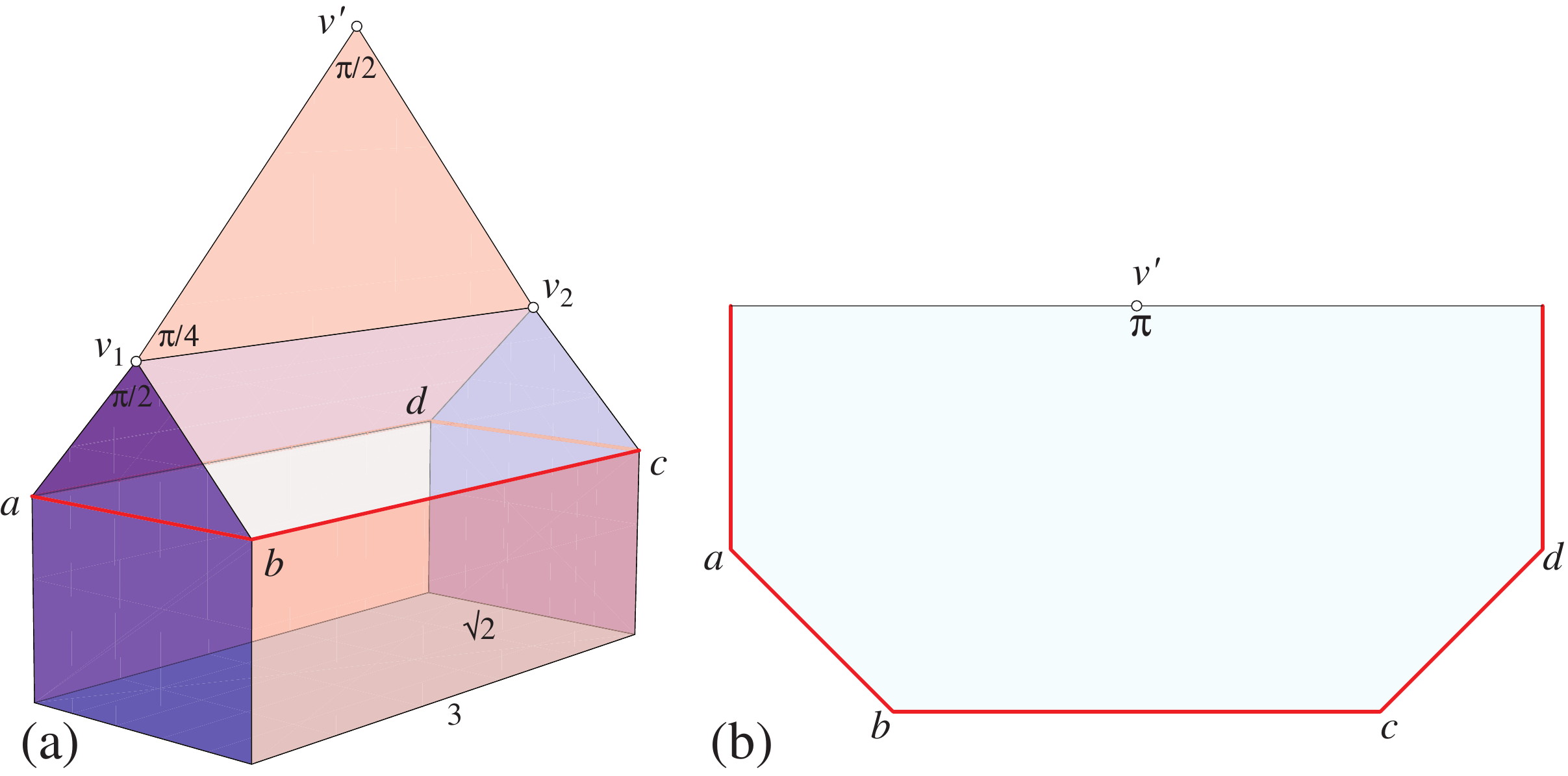}
\caption{(a)~$C=(a,b,c,d)$ is a convex curve with
angle $\frac{3}{4}\pi$ to the left at each vertex.
The curvature at $v_1$ and at $v_2$ is $\frac{1}{2}\pi$.
(b)~Cutting along the generator from $v'$ through the midpoint of $ad$ and
developing $C$ shows that it lives on a cone
with apex angle $\pi$ at $v'$.
(Base of $\P$ is $3 \times \sqrt{2}$.)
}
\figlab{VertexMerging}
\end{figure}

\paragraph{Half-Surfaces Notation.}
$C$ partitions $\P$ into two \emph{half-surfaces}:
$\P \setminus C$.  
We call the left and right half-surfaces $P_L$ and $P_R$ respectively,
or $P$ if the distinction is irrelevant.
We view each half-surface as closed, with boundary $C$.

\section{Convex Curves}
\seclab{Convex}
We start with convex curves $C$.

\paragraph{Convexity of Half-Surfaces.}
In order to apply vertex merging, we use a lemma to guarantee
the existence of a pair to merge.
We first remark that it is not the case that every half-surface
$P \subset \P$ bounded by a convex curve $C$ is \emph{convex}
in the sense that, if $x,y \in P$, 
then a shortest path $\g$ of $\P$ connecting $x$ and $y$ lies 
in $P$.  

\noindent\emph{Example.}
Let $\P$ be defined as follows.  Start with
the top half of a regular octahedron, whose four equilateral
triangle faces form a pyramid over a square base $abcd$.  
Flex the pyramid by squeezing $a$ toward $c$ slightly while maintaining the four equilateral triangles,
a motion 
which separates $b$ from $d$.
Define $\P$ to be the convex hull of these
four moved points $a'b'c'd'$ and the pyramid apex.
Let $C=(a',b',c',d')$ and let $P$ be the half-surface including the 
four equilateral triangles.  Then $a'$ and $c'$ are in $P$,
but the edge $a'c'$ of $\P$, which is the shortest path connecting
those points, is not in $P$: it crosses the ``bottom'' of $\P$.

Although $P$ may not be convex, $P$ is \emph{relatively convex} in the sense
that it is isometric to a convex half-surface:
there is some $\P^\#$ and a half-surface $P^\# \subset \P^\#$ such
that $P$ is isometric to $P^\#$ and $P^\#$ is convex.

\begin{lemma}
Every half-surface $P \subset \P$ bounded by a convex curve $C$ is
relatively convex,
i.e., $P$ is isometric to a half-surface that
contains a shortest path $\g$ between any two of its points $x$ and $y$.
More particularly, if neither $x$ nor $y$ is on $C$, then 
the shortest path $\g$
contains no points of $C$.
If exactly one of $x$ or $y$ is on $C$, then that is the only point
of $\g$ on $C$. 
\lemlab{vve}
\end{lemma}
\begin{proof}
We glue two copies of $P$ along $\bP=C$.
Because $C$ is convex, Alexandrov's Gluing Theorem says
the resulting surface is isometric to a unique polyhedral surface, call it $\P^\#$.
Because $\P^\#$ has intrinsic symmetry with respect to $C$,
a lemma of
Alexandrov~\cite[p.~214]{a-cp-05} applies to show that the
polyhedron $\P^\#$ has a symmetry plane $\Pi$ containing $C$.

Now consider the points $x$ and $y$ in the upper half $P$ of $P^\#$,
at or above $\Pi$.
If $\g$ is a shortest path from $x$ to $y$, then by the symmetry of $\P^\#$,
so is its reflection $\g'$ in $\Pi$.
Because shortest paths on convex surfaces do not branch,
$\g$ must lie in the closed half-space above $\Pi$, and so lies on $P$.

If neither $x$ nor $y$ are on $C$, they are strictly above $\Pi$, and $\g$ must be
as well to avoid a shortest-path branch.  
If, say, $x \in C$ but $y \not\in C$, and if $\g$ touched $C$ elsewhere,
say at $z$, then 
from $y$ to $x$ we have a shortest path $\gamma$ and 
another shortest path, composed of
the arc of $\g$ from $y$ to $z$ and
the arc of $\g'$ from $z$ to $x$, hence
we would have a shortest-path branch at $z$.
If both $x$ and $y$ are on $C$, then either $\g$ meets $C$ in exactly
those two points, or $\g \subset C$, for the same reason as above.
\end{proof}

\begin{lemma}
Let $C$ be a convex curve on $\P$, convex to its left.
Then $C$ lives on a cone $\L_L$ to its left side,
whose apex $a$ has curvature $\O_L$.
\lemlab{ConvexCone}
\end{lemma}
\begin{proof}
By the Gauss-Bonnet theorem, $\tau_L + \O_L = 2 \pi$.
Because $\tau_L \ge 0$ for a convex curve,
we must have $\O_L \le 2 \pi$.
Let $V$ be the set of vertices of the half-surface $P_L$ not on $C$.

Suppose first that $\O_L < 2 \pi$. 
If $|V|=1$,
then $P_L$ is a pyramid, which is already a cone.
So suppose $|V| \ge 2$, and let $v_1$ and $v_2$ be any two vertices in $V$.
Lemma~\lemref{vve} guarantees that a shortest path $\g$ between them is in $P_L^\#$
and disjoint from $C$. 
Perform vertex merging along $\g$, resulting in a new vertex $v'$
whose curvature is the sum of that of $v_1$ and $v_2$.  Note
that merging is always possible, because $\o_1 + \o_2 \le \O_L < 2\pi$.
Also note that $v'$ is not on $C$, by Lemma~\lemref{vve}.
Let $N_L$ be some small left neighborhood of $C$ in $P_L$.
Then $N_L$ is unaffected by the vertex merging:
neither $v_1$ nor $v_2$ is in $N_L$ because it is vertex free,
and $N_L$ may be chosen narrow enough (by Lemma~\lemref{vve})
so that no portion of $\g$ is in $N_L$.
Replace $V$ by $(V \setminus \{v_1,v_2\}) \cup \{v'\}$.
 
Continue vertex merging in a like manner between vertices of $V$
until $|V|=1$, 
at which point we have $C$ and $N_L$ living on a cone, as claimed.

If $\O_L = 2\pi$, then 
the last step of vertex merging will not succeed.
However, we can see that a slight altering of the two
glued triangles so that
$\O_L < 2 \pi$
will result in the cone apex approaching infinity,
as follows.  Cut along a geodesic between the two vertices,
say $v_i$ and $v_{i+1}$, and insert double triangles of
base angles 
$\frac{1}{2}\o_i$ and respectively $\frac{1}{2}\o_{i+1} - \e_n$,
with $\e_n > 0$ and $\lim _n \e_n =0$.
And so in this case, $C$ and $N_L$ live on a cylinder, which we
consider a degenerate cone.
\end{proof}

\noindent\emph{Example.}
In Figure~\figref{VertexMerging} the two vertices inside $C$,
of curvature $\frac{1}{2}\pi$ each, are merged to
one of curvature $\pi$,
which is then the apex of a cone on which $C$ lives.

\noindent\emph{Example.}
Figure~\figref{ThreeFlat}(a) shows an example with three
vertices inside $C$.
$\P$ is a doubly covered flat pentagon, and
$C=(v_4,v_5,v_4)$ is the
closed curve consisting of a repetition of the segment $v_4 v_5$.
$C$ has $\pi$ surface angle at every point to its left, 
and so is convex. 
The curvatures at the other vertices
are $\o_1=\pi$ and $\o_2=\o_3=\frac{1}{2}\pi$.
Thus $\O_L=2\pi$,
and the proof of Lemma~\lemref{ConvexCone} shows
that $C$ lives on a cylinder.
Following the proof,
merging $v_1$ and $v_2$ removes those vertices and creates
a new vertex $v_{12}$ of curvature $\frac{3}{2}\pi$;
see~(b) of the figure.
Finally merging $v_{12}$ with $v_3$ creates a 
``vertex at infinity'' $v_{123}$ of curvature
$2\pi$.  Thus $C$ lives on a cylinder as claimed.
If we first merged $v_2$ and $v_3$ to $v_{23}$, and then $v_{23}$ to $v_1$,
the result is exactly the same, although less obviously so.
\begin{figure}[htbp]
\centering
\includegraphics[width=\linewidth]{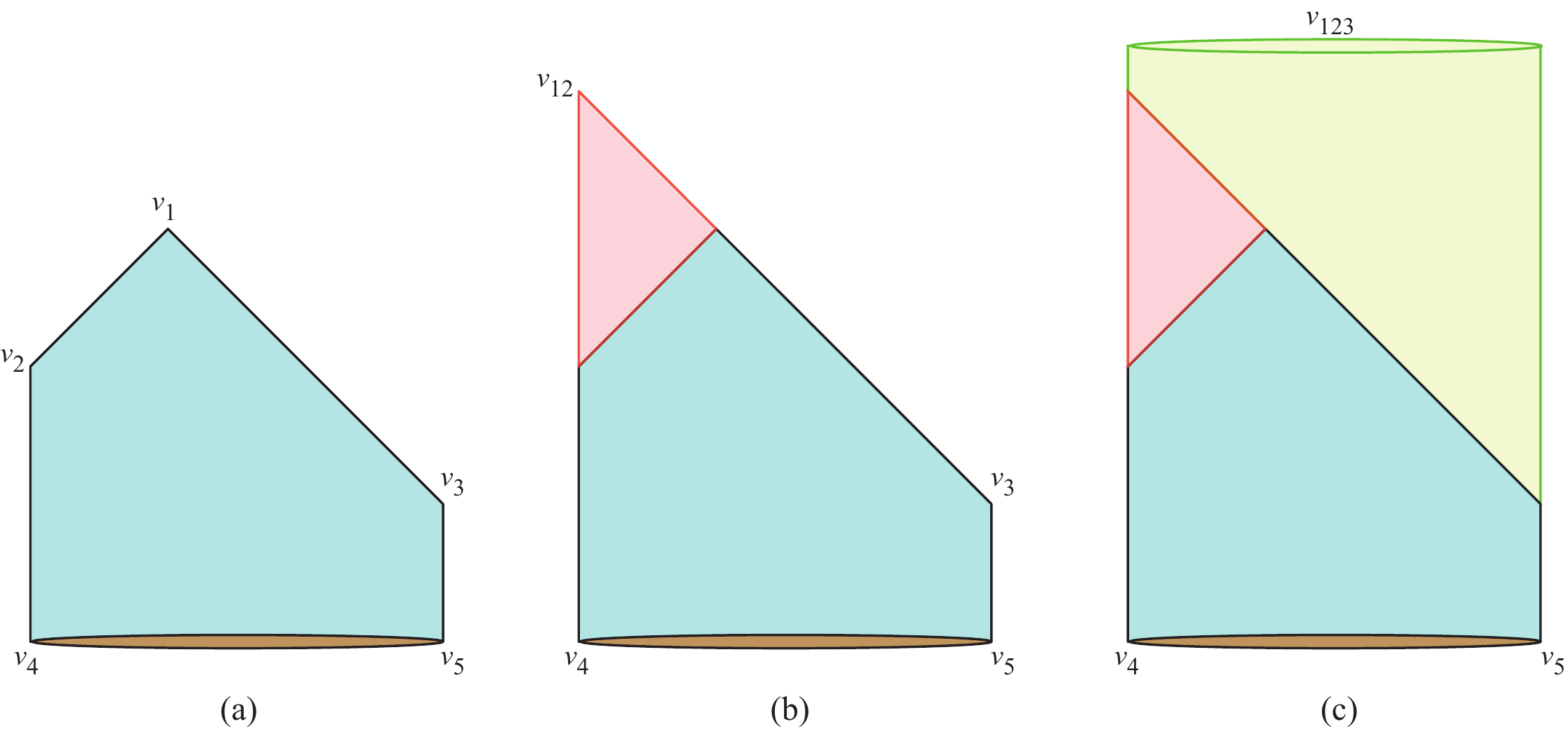}
\caption{(a)~A doubly covered flat pentagon.
(b)~After merging $v_1$ and $v_2$.
(c)~After merging $v_{12}$ and $v_3$.
}
\figlab{ThreeFlat}
\end{figure}

This last example raises the natural question of whether the cone
constructed through vertex merging in
Lemma~\lemref{ConvexCone}
is independent of the order of merging.
Indeed the determined cone is unique:

\begin{lemma}
A curve $C$ that lives on a cone $\L$ (say, to its left) 
uniquely determines that cone.
\lemlab{ConeUnique}
\end{lemma}
\begin{proof}
Suppose that $C$ lives on two cones $\L$ and $\L'$.
We will show that the regions of these two cones bounded by $C$
are isometric.
First note that the apex angle of both $\L$ and $\L'$ is $\O_L$,
the total curvature inside and left of $C$.
Let $x \in C$ be a point of $C$ that has a tangent $t$ to one side, and let $x_1$ be a
point in the plane
and $t_1$ a direction vector from $x_1$.
Roll $\L$ in the plane so that $x$ and $t$ coincide with $x_1$ and
$t_1$.
Continue rolling until $x$ is encountered again; call that point of
the plane $x_2$.
The resulting positions of $x_1$ and $x_2$ are 
the same as
would be produced by cutting the cone along a generator $ax$.

If $x_1=x_2$, then both $\L$ and $\L'$ are planar and so isometric.
So assume $x_1 \neq x_2$.
If $\O_L \ge \pi$, then the cone angle $\a \le \pi$, as in
Figure~\figref{ConeDevelopment}(b).
The segment $x_1 x_2$ determines two isosceles triangles
with apex angle $\a$, only one of which can correspond to the left
side of $\overline{C}$.
\begin{figure}[htbp]
\centering
\includegraphics[width=0.95\linewidth]{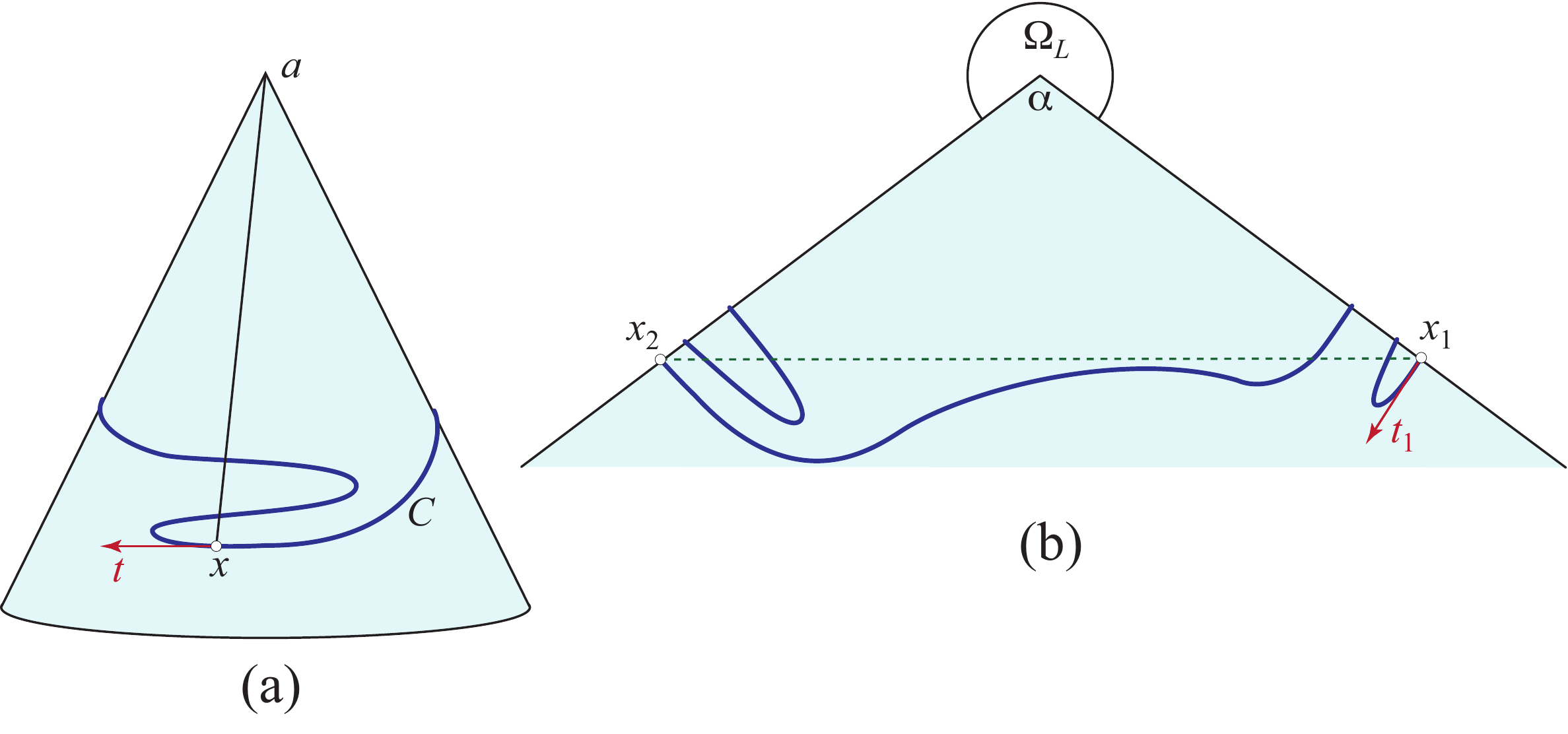}
\caption{(a)~Cone $\L$ on which $C$ lives.
(b)~Positions of $x_1$ and $x_2$ after cutting open $\L$ along $ax$.}
\figlab{ConeDevelopment}
\end{figure}
Analogously, if $\O_L < \pi$, then $x_1 x_2$ determines a unique
isosceles triangle of apex angle $\O_L$, 
the equal sides of which bound, together with $\overline{C}$,
the region of $\L$
to the left
of $\overline{C}$.
Note that $\overline{C}$ doesn't actually depend on
the cones $\L$ and $\L'$, but only on the left neighborhood of $C$ in
$P$, 
and hence this development is the same for $\L$ and $\L'$.
So, up to planar isometries, the planar unfolding of the cone
supporting $C$ is unique, and thus the cone itself
and the position of $C$ on it are unique up to isometries.
\end{proof}

\noindent
Note that this lemma does not assume that $C$ is convex;
rather it holds for any closed curve $C$.


Finally we establish the visibility property mentioned in the introduction.

\begin{lemma}
A convex curve $C$ on $\P$
is visible from the apex $a$ of the unique cone $\L$
on which it lives to its convex side. 
\lemlab{ConvexVisible}
\end{lemma}
\begin{proof}
With $C$ directed so that its convex side is its left side,
which we may consider its interior,
the apex $a$ is inside $C$.
Assume there is a cone generator intersecting $C$ twice.
Then, rotating the generator around the apex in one direction
or the other eventually must reach
a generator $ax$ tangent to $C$ at $x$ where
$L(x) > \pi$,
contradicting convexity.
See Figure~\figref{ConvexSideVisibility}.
\begin{figure}[htbp]
\centering
\includegraphics[width=0.6\linewidth]{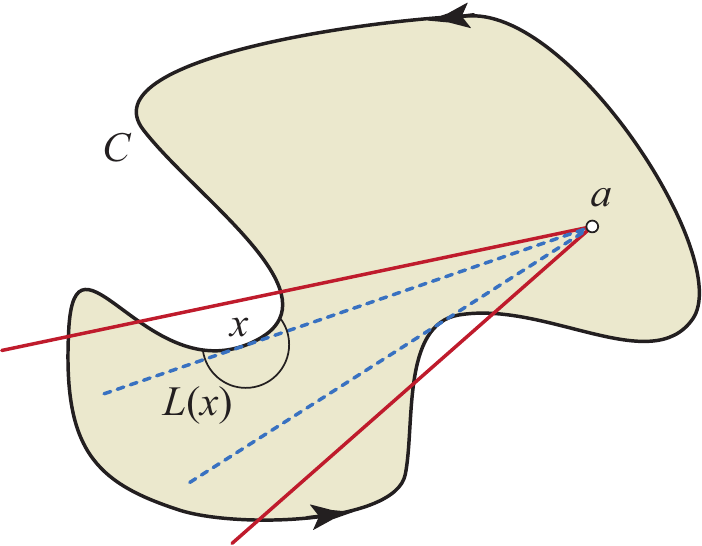}
\caption{No generator may cross $C$ twice.  
}
\figlab{ConvexSideVisibility}
\end{figure}
\end{proof}

\noindent
This lemma may as well be established with a different proof, whose sketch is
as follows.  Let $z$ be the closest point of $C$ to $a$.  Then $az$
must be orthogonal to $C$ at $p$. Inserting a ``curvature triangle''
along $az$ with apex angle $\o(a)$ flattens $P$ to a planar domain
with a convex boundary, and visibility from $a$ follows.


We gather the previous three lemmas into a summarizing theorem:
\begin{theorem}
Any curve $C$, convex to its left, lives on a unique cone $\L_L$
to its left side.
$\L_L$ has curvature $\O_L$ at its apex, and so has  apex angle
$2\pi-\O_L$.
Every point of $C$ is visible from the cone apex $a$.
\thmlab{Convex}
\end{theorem}

\section{Convex Loops}
\seclab{ConvexLoops}
Consider the polyhedron $\P$ shown in
Figure~\figref{ConvexLoop}(a), which is a variation
on the example from Figure~\figref{VertexMerging}(a).
Here $C=(a,b,b',x,c',c,d)$ is a convex loop, with loop point $x$.
The cone on which it should live is analogous to 
Figure~\figref{VertexMerging}(b):
vertex merging of $v_1$ and $v_2$ again produces the cone
apex $v'$ whose curvature is $\pi$.
But $C$ does not ``fit'' on this cone, as
Figure~\figref{ConvexLoop}(b) shows;
the apex $a=v'$ is not inside $C$.
\begin{figure}[htbp]
\centering
\includegraphics[width=\linewidth]{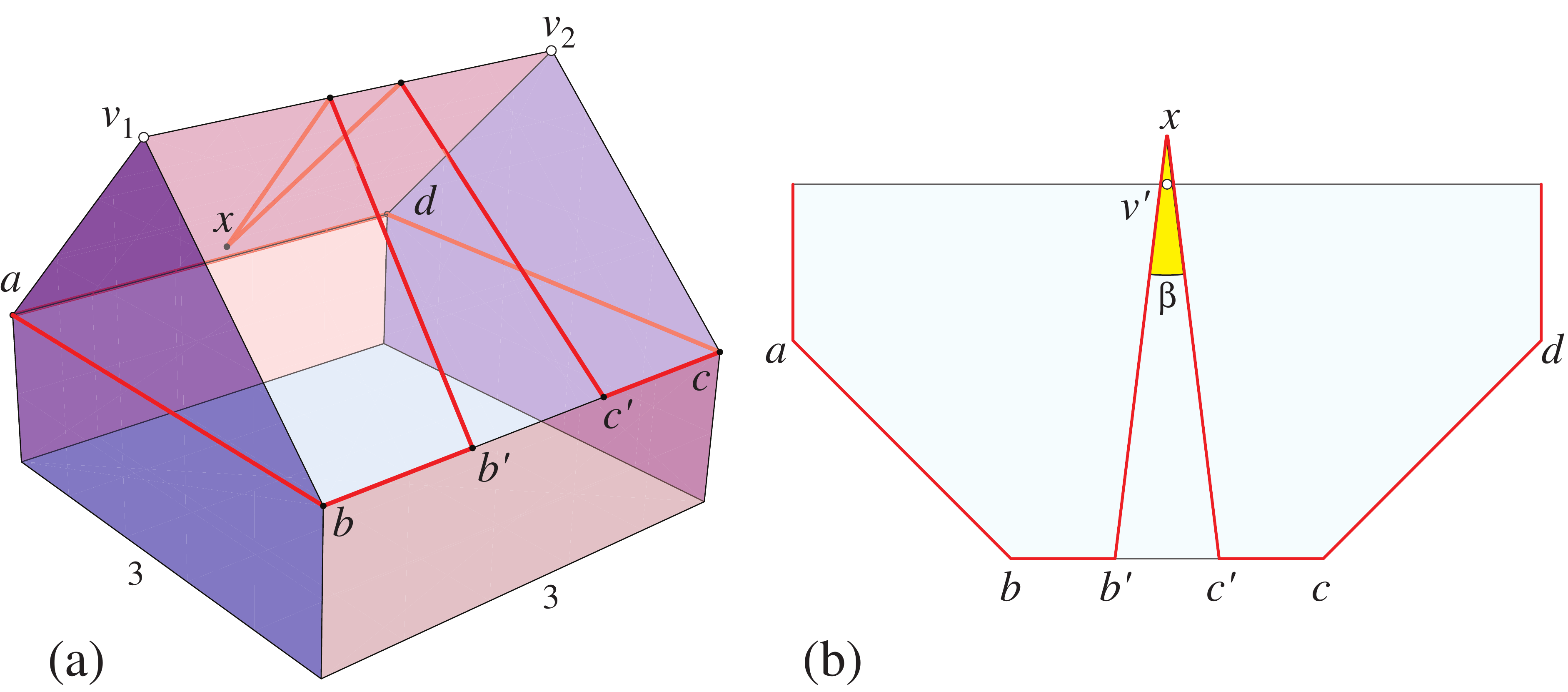}
\caption{(a)~A convex loop $C$ that does not live on a cone.
(b)~A flattening of the cone on which it should live.
(Base of $\P$ is $3 \times 3$.)
}
\figlab{ConvexLoop}
\end{figure}

We remark that, if the central ``spike'' $(b',x,c')$ is
shortened, it does live on the cone.
Even for convex loops that do live on a cone, there are
examples that fail
to satisfy the visibility property,
Lemma~\lemref{ConvexVisible}.
Simply shifting the spike in this example
to one side of $v'$ blocks visibility to portions of $C$.


\section{Reflex Curves and Reflex Loops}
\seclab{Reflex}

Recall that, for each corner $c_i$ of a curve $C$,
$\a_i + \o_i + \b_i = 2 \pi$, where $\a_i$ and $\b_i$ are the left and
right
angles at $c_i$ respectively, and $\o_i$ is the Gaussian curvature at
$c_i$.
When $C$ is vertex-free, $\o_i=0$ at all corners, and the
relationships
among the curve classes is simple and natural:
the other side of a convex curve is reflex, the other side of a reflex
curve is convex.
The same holds for the loop versions: the other side of a convex loop
is a reflex loop (because $\a_m \ge \pi$ implies $\b_m \le \pi$, where
$c_m$ is the loop point), and the other side of a reflex loop is a
convex loop.
When $C$ includes vertices, the relationships between the curve
classes
is more complicated.
The other side of a convex curve is reflex only if the curvatures
at the vertices on $C$ are small enough so that $\a_i + \o_i \le \pi$;
$C$ would still be convex even if it just included those vertices
inside.
The same holds for convex loops, as summarized in the table below.

On the other hand, the other side of a reflex curve is always convex,
because nonzero vertex curvatures only make the other side more convex.
The other side of a reflex loop is a convex loop, and it is a
convex curve if the curvature at the loop point $c_m$ is large enough
to force $\a_m \le \pi$, i.e., if $\b_m + \o_m \ge \pi$.

\begin{table}[htbp]
\begin{center}
\begin{tabular}{| c | c  |}
        \hline
\emph{Curve class} & \emph{Other side, and condition}
         \\ \hline \hline
convex & reflex only if $\forall i$, $\a_i + \o_i \le \pi$ 
        \\ \hline
convex loop & reflex loop only if $\forall i
\neq m$, $\a_i + \o_i \le \pi$ (necessarily, $\b_m \le \pi$)
        \\ \hline
reflex & convex (always)
        \\ \hline
reflex loop & convex loop (always), and convex if $\b_m + \o_m \ge \pi$
        \\ \hline
\end{tabular}
\caption{Other-side conditions for curve classes.
$m$ indexes the loop-point corner $c_m$ for loop versions.}
\tablab{Reflex}
\end{center}
\end{table}

This latter subclass of reflex loops---those whose other side is
convex---especially interest us, because any convex curve that
includes
at most one vertex is a reflex loop of that type.
All our results in this section hold for this class of curves.



\begin{lemma}
Let $C$ be a curve that is either reflex (to its right), 
or a reflex loop which is convex to
the other (left) side, 
with $\b_m < \pi$ at the loop point $c_m$.
Then $C$ lives on a cone $\L_R$ to its reflex side.
\lemlab{ReflexCone}
\end{lemma}
\begin{proof}
Again let $c_1,c_2,\ldots,c_m$ be the corners of $C$, 
with $c_m$ the loop point if $C$ is a reflex loop.
Because $C$ is convex to its left,
we have $\O_L \le 2 \pi$.
Just as in Lemma~\lemref{ConvexCone},
merge the vertices strictly in $P_L$ to one vertex $a$.
Let $\L_L$ be the cone with apex $a$ on which $C$ now lives.
It will simplify subsequent notation to let $\L = \L_L$.

\begin{figure}[htbp]
\centering
\includegraphics[width=0.75\linewidth]{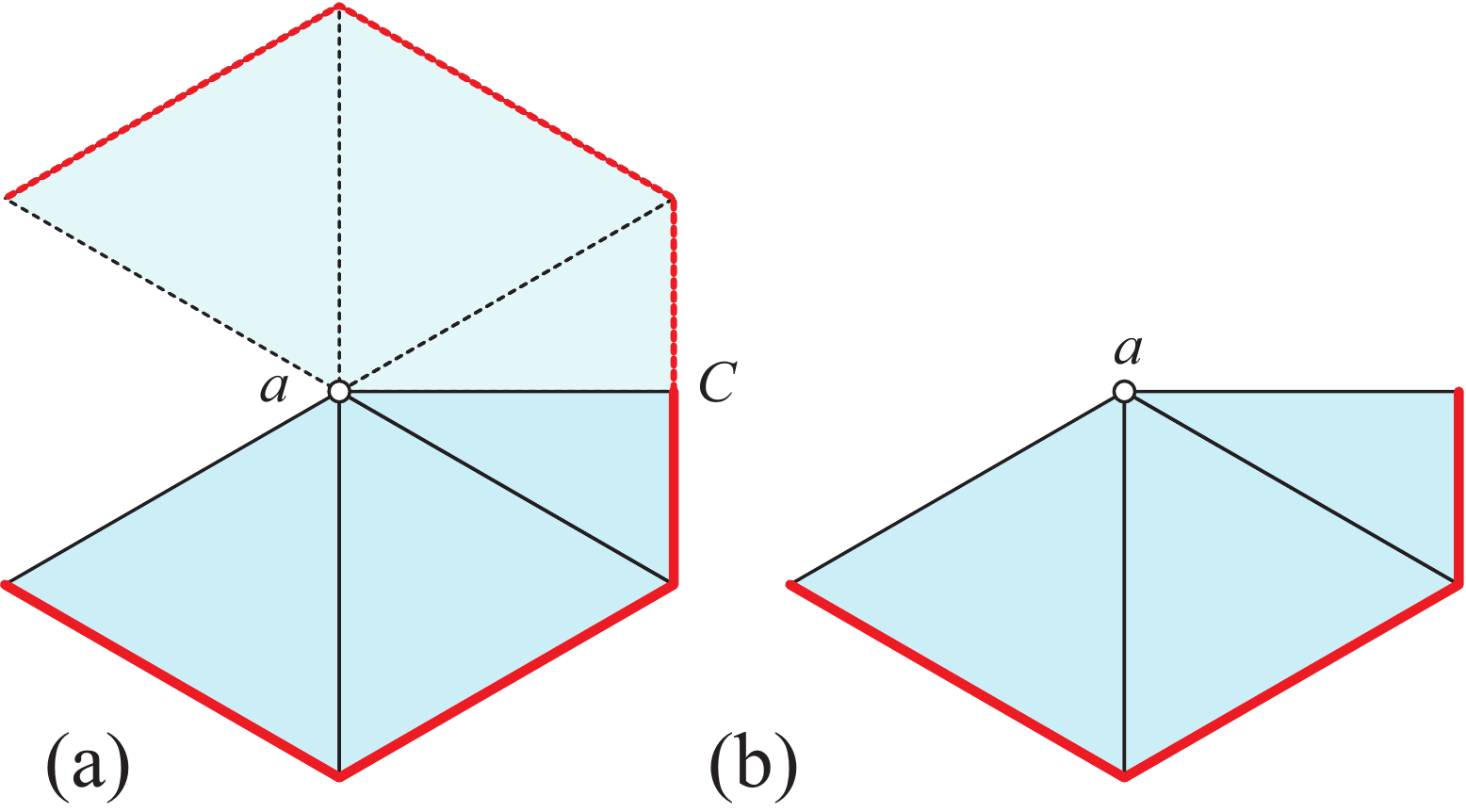}
\caption{The cone $\L$ for $C$ in Figure~\protect\figref{IcosaCuboctaCurves}(a),
opened~(a) and doubly covered~(b).
}
\figlab{IcosaDouble0}
\end{figure}

Let $N_R$ be a (small) right neighborhood of $C$, a neighborhood to
the reflex side of $C$.  
For subsequent subscript embellishment, we use $N$ to represent $N_R$.
Its shape is irrelevant to
the proof, as long as it is vertex free and its left boundary
is $C$.

Join $a$ to each corner $c_i$ by a cone-generator $g_i$
(a ray from $a$ on $\L$). Lemma~\lemref{ConvexVisible} ensures this is
possible.
Cut along $g_i$ beyond $c_i$ into $N$.  There are choices how to extend
$g_i$ beyond $c_i$, but the choice does not matter for our purposes.
For example, one could choose a cut that bisects $\b_i$ at $c_i$.
Insert along each cut into $N$ a \emph{curvature triangle}, that is,
an isosceles triangle with two sides equal to the cut length, and apex angle
$\o_i$ at $c_i$.
(If $c_i$ does not coincide with a vertex of $\P$, then $\o_i=0$ and
no curvature triangle is inserted.)
This flattens the surface at $c_i$, and ``fattens'' $N$ to $N'$ without altering
$C$ or the cone $\L$ up to $C$.
Now $N'$ lives on the same cone $\L$ that $C$ and its left neighborhood $N_L$ do.

From now on we view $\L$ and the subsequent cones we will construct
as flattened into the plane, producing a doubly covered cone with
half the apex angle.
(Notice that here
``doubly covered'' above refers to a neighborhood of the cone apex, 
and not to the image of the curve $C$.)
It is always possible to choose any generator
$ax$ for $x \in C$ and flatten so that $ax$ is the leftmost extreme
edge of the double cone.
We start by selecting $x=c_1$, so that $g_1$ is the leftmost extreme;
let $h_1$ be the rightmost extreme edge.
We pause to illustrate the construction before proceeding.

\begin{figure}[htbp]
\centering
\includegraphics[width=0.75\linewidth]{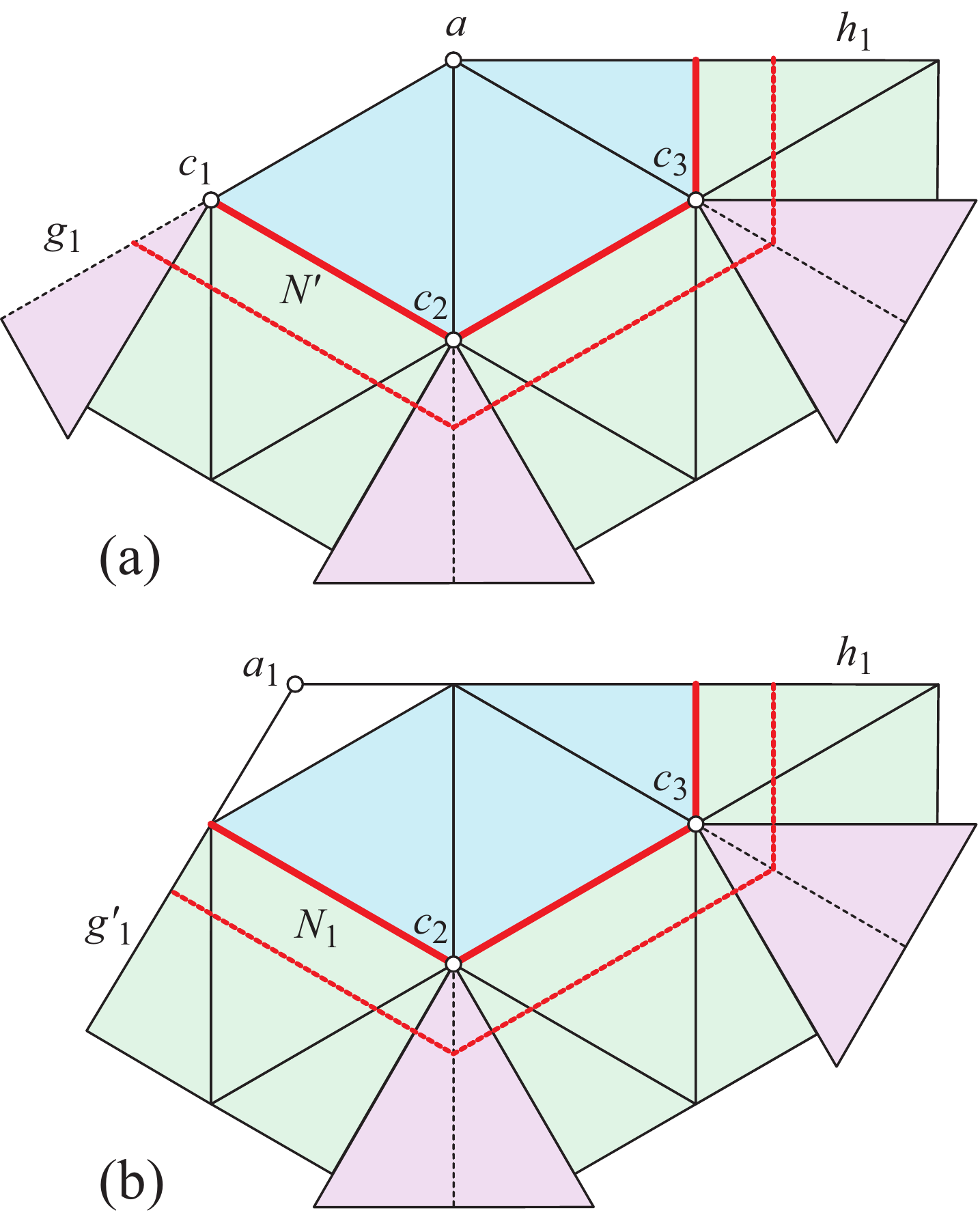}
\caption{(a)~After insertion of curvature triangles, $N'$ lives on $\L$.
(b)~Removing the doubly covered half curvature triangle at $c_1$
leads to a new cone $\L_1$.
(In this and in Figure~\protect\figref{IcosaDouble2} we
display the full icosahedron faces to the right of $C$,
although only a small neighborhood is relevant to the proof.)
}
\figlab{IcosaDouble1}
\end{figure}

Let $C$ be the curve on the icosahedron illustrated in 
Figure~\figref{IcosaCuboctaCurves}(a).
This curve already lives on the cone $\L$ without any vertex merging.
Figure~\figref{IcosaDouble0}(a) shows the five equilateral triangles
incident to the apex, and~(b) shows the corresponding doubly covered cone.
Figure~\figref{IcosaDouble1}(a) illustrates $\L$ after insertion of
the curvature triangles, each with apex angle $\o_i=\frac{1}{3}\pi$.
A possible neighborhood $N'$ is outlined.

After insertion of all curvature triangles, we in some sense erase
where they were inserted, and just treat $N'$ as a band living on
$\L$.
Now, with $g_1$ the leftmost extreme, we identify a half-curvature
triangle on the front side, matched by a half-curvature triangle on
the back side, incident to $c_1$ in $N'$. Each triangle has angle $\frac{1}{2}\o_1$
at $c_1$.
See again Figure~\figref{IcosaDouble1}(a).
Now rotate $g_1$ counterclockwise about $c_1$ by $\frac{1}{2}\o_i$, 
and cut out the two half-curvature triangles from $N'$,
regluing the front to the back along the cut segment.
Extend the rotated line $g'_1$ to meet the extension of $h_1$.
Their intersection point is the apex $a_1$ of a new (doubly covered) cone $\L_1$,
on which neither $a$ nor $c_1$ are vertices.
Note that  the rotation of $g_1$ effectively
removes an angle of measure $\o_1$ incident to $c_1$ from the $N'$
side,
and inserts it
on the other side of $C$. 
See Figure~\figref{IcosaDouble1}(b).
Call the new neighborhood $N_1$, and the new convex curve $C_1$.
$C_1$ is the same as $C$ except that the angle at $c_1$ is now
$\a_1 + \o_1$, which by the assumption of the lemma, is still convex
because $\b_1 \ge \pi$.

Now we argue that $g'_1$ does not intersect $N_1$ other than where it forms
the leftmost boundary.  For if $g'_1$ intersected $N_1$ elsewhere, then,
taking $N_1$ to be smaller and smaller, tending to $C_1$, we conclude
that $g'_1$ must intersect $C_1$ at a point other than $c_1$.  
But this contradicts the fact
that 
either of the two planar images (from the two sides of $\L$) 
of $C_1$ is convex.
Indeed $g'_1$ is a supporting line at $c_1$ to the convex set constituted by
$\L_1$ up to $C_1$.

Note that we have effectively merged vertices $c_1$ and $a$ to form $a_1$,
in a manner similar to the vertex merging used in Lemma~\lemref{ConvexCone}.
The advantage of the process just described is that it does not rely
on having a triangle half-angle no more than $\pi$ at the new cone apex.

\begin{figure}[htbp]
\centering
\includegraphics[width=0.75\linewidth]{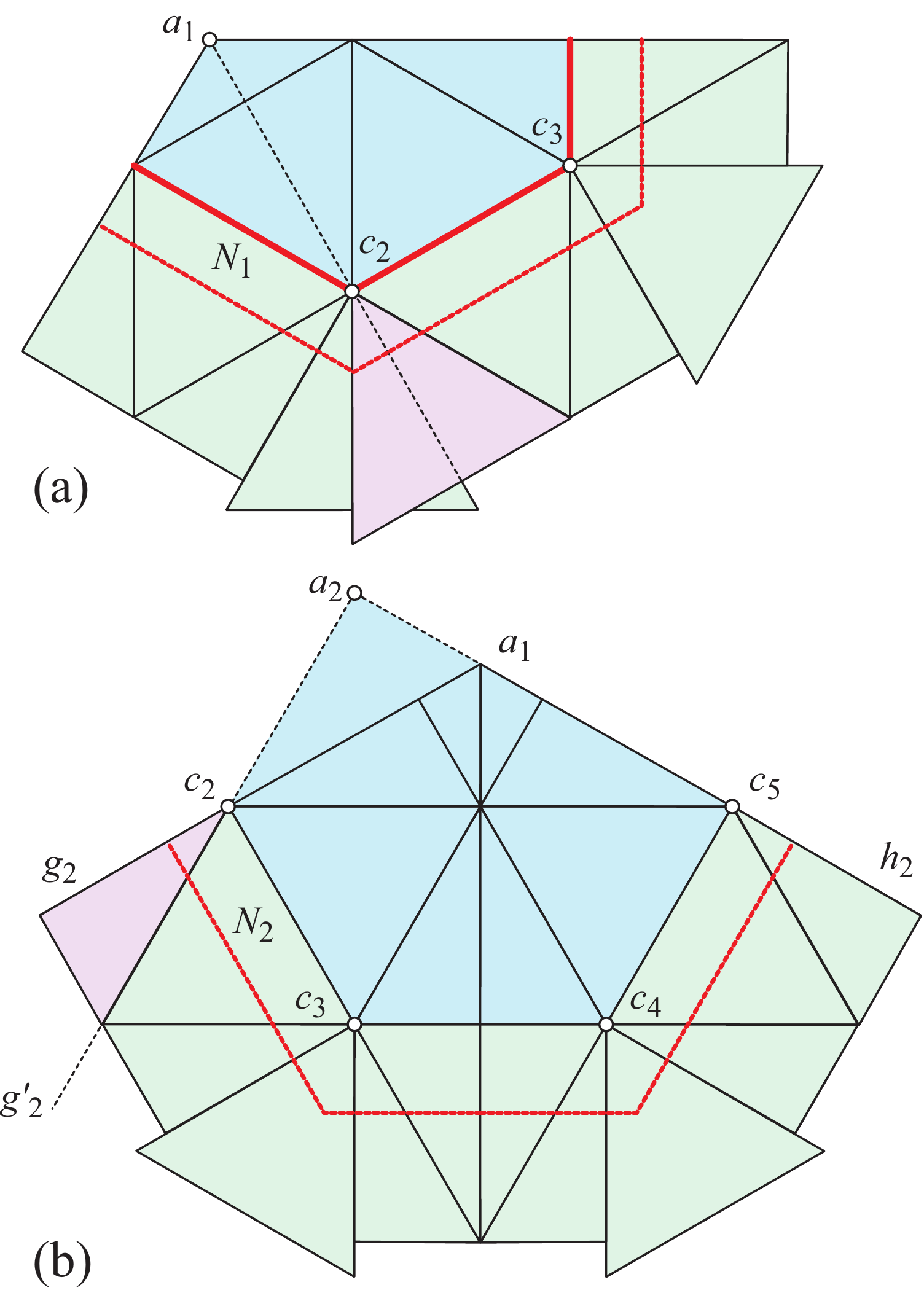}
\caption{(a)~Generator $g_2$ from $a_1$ through $c_2$ into $N_1$.
(b)~Reoriented so $g_2$ is left extreme.
}
\figlab{IcosaDouble2}
\end{figure}

Next we eliminate the curvature triangle inserted at $c_2$.  Let $g_2$ be the generator from $a_1$ through $c_2$
(again, Lemma~\lemref{ConvexVisible} applies).
Identify a curvature triangle of apex angle $\o_2$ in $N_1$ bisected
by $g_2$;
see Figure~\figref{IcosaDouble2}(a).
Now reflatten the cone $\L_1$ so that $g_2$ is the left extreme,
and let $h_2$ be the right extreme, as in~(b) of the figure.
Rotate $g_2$ by $\frac{1}{2}\o_2$ about $c_2$ to produce $g'_2$, cut out
the half-curvature triangles on both the front and back of $N_1$,
and extend $g'_2$ to meet the extension of $h_2$ at a new apex $a_2$.
Now we have a new neighborhood $N_2$, with left
boundary the convex curve $C_2$, living on a cone $\L_2$.

We apply this process through $c_1,\ldots,c_{m-1}$.
It could happen at some stage that $g'_i$ and the $h_i$ extension meet
on the other side of $C_i$, in which case the cone apex is to the reflex
side.
(Or, they could be parallel and meet ``at infinity,'' which is what occurs with
the icosahedron example.)
From the assumption of the lemma that $\b_i \ge \pi$
for $i < m$, $\a_i + \o_i \le \pi$ and so the curves $C_i$ remain convex throughout
the process.
So the argument above holds.

For the last, possibly exceptional corner $c_m$, $C_{m-1}$ 
from the previous step is convex,
but the final step could render $C_m$ nonconvex
(if $\a_m + \o_m > \pi$). 
But as there is
no further processing, this nonconvexity does not affect the proof.
\end{proof}

\medskip
\noindent
For the icosahedron example, five insertions of $\frac{1}{3}$ curvature
triangles, together with the original $\frac{1}{3}$ curvature at $a$,
produces a cylinder.  And indeed, $\b_i=\pi$ for the five $c_i$
corners of $C$, and $C$ forms a circle on a cylinder.


\begin{lemma}
Let $C$ be a curve satisfying the same conditions
as for Lemma~\lemref{ReflexCone}.
Then $C$ is visible from the apex $a$ of the cone $\L$
on which it lives to its reflex side. 
\lemlab{ReflexVisible}
\end{lemma}
\begin{proof}
Again letting $c_1,\ldots,c_m$ be the corners of $C$, with
$c_m$ the possibly exceptional vertex,
we know that $\b_i \ge \pi$ for $i=1,\ldots,m{-}1$,
but it may be that $\b_m < \pi$.
Just as in the proof of Lemma~\lemref{ReflexCone},
we flatten $\L$ into the plane, this time choosing $c_m$ to lie
on the leftmost extreme  generator $L_1$of $\L$.  Let $b$ be the point of $C$ that
lies on the rightmost extreme generator $L_2$ 
in this flattening.
Finally, let $C_u$ be the portion of $C$ on the upper surface of the
flattened $\L$, and $C_l$ the portion on the lower surface.
See Figure~\figref{VisibilityNonConvex}.
\begin{figure}[htbp]
\centering
\includegraphics[width=0.75\linewidth]{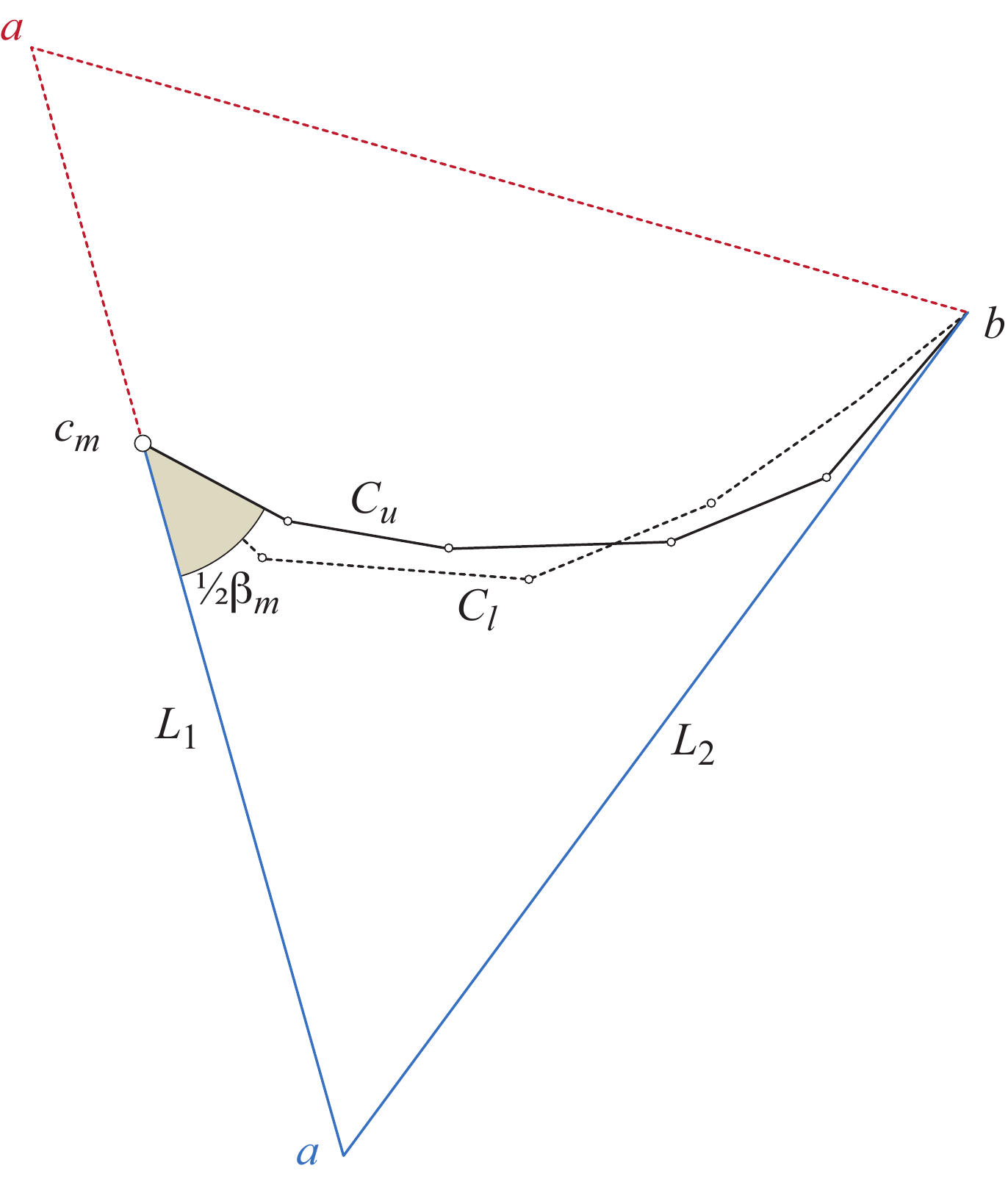}
\caption{The apex $a$ could lie either to the reflex or to the convex side of $C$.}
\figlab{VisibilityNonConvex}
\end{figure}
Now that we have placed the one anomalous corner on the extreme boundary $L_1$,
both $C_u$ and $C_l$ present a uniform aspect to the apex $a$,
whether it is to the convex or reflex side of $C$:
every corner of $C_u$ and $C_l$ is reflex (or flat) toward the
reflex side, and convex (or flat) toward the convex side.
In particular, $c_m b \cup C_u$ is a planar convex domain.
Each line through $a$ intersects $c_m b$ exactly once,
and therefore intersects $C_u$ exactly once; and similarly for $C_l$.
\end{proof}


Just as we observed for convex loops, this visibility lemma does
not hold for all reflex loops---the assumption that the other side
is convex is essential to the proof.

We summarize this section in a theorem (recall that $\O_L +\O_C + \O_R = 4\pi$).
\begin{theorem}
A curve $C$ that is either reflex (to its right), or a reflex loop which is convex to
the other (left) side, 
lives on a unique cone $\L_R$
to its reflex side.
If $\O_R > 2\pi$, then the reflex neighborhood $N_R$
is to the unbounded side of $\L_R$, i.e., the apex of $\L_R$ is left
of $C$;
if $\O_R < 2\pi$, then $N_R$ is to the bounded side, i.e.,
the apex of $\L_R$ is to the right side of $C$.
If $\O_R=2\pi$, $C \cup N_R$ lives on a cylinder.
In all cases, every point of $C$ is visible from the cone apex $a$.
\thmlab{Reflex}
\end{theorem}
\begin{proof}
The uniqueness follows from Lemma~\lemref{ConeUnique}.
The cone $\L_R$ constructed in the proof of Lemma~\lemref{ReflexCone}
results in the cone apex to the convex side of $C$ as long as
$\O_L + \O_C \le 2 \pi$, when $\O_R \ge 2 \pi$.
Excluding the cylinder cases, this justifies the claims
concerning on which side of 
$\L_R$ the neighborhood $N_R$ resides.
The apex curvature of $\L_R$ is $\min \{ \O_L + \O_C, \O_R \}$.
\end{proof}


\emph{Example.}
An example of a reflex loop that satisfies the hypotheses of Theorem~\thmref{Reflex} is shown in
Figure~\figref{CubeOctCurve2}(a).
Here $C$ has five corners, and is convex to one side at each.
$C$ passes through only one vertex of the cuboctahedron $\P$,
and so it is reflex at the four non-vertex corners
to its other side.
Corner $c_5$ coincides with a vertex of $\P$, which has 
curvature $\o_5=\frac{1}{3}\pi$.
Here $\a_5 = \b_5 = \frac{5}{6}\pi$.
Because $\b_5 < \pi$, $C$ is a reflex loop.
We have $\O_L = \frac{2}{3}\pi$ because $C$ includes two cuboctahedron
vertices, $u$ and $v$ in the figure.
$\O_C= \o_5 = \frac{1}{3}\pi$.
And therefore $\O_R = 3\pi$.
The apex curvature of $\L_L$ is $\O_L = \frac{2}{3}\pi$,
and the apex curvature of $\L_R$ is
$\min \{ \O_L + \O_C, \O_R \} = \pi$.
$N_R$ lives on the unbounded side of this cone, 
which is shown shaded in
Figure~\figref{CubeOctCurve2}(b).
Note the apex $a$ is left of $C$.
\begin{figure}[htbp]
\centering
\includegraphics[width=0.9\linewidth]{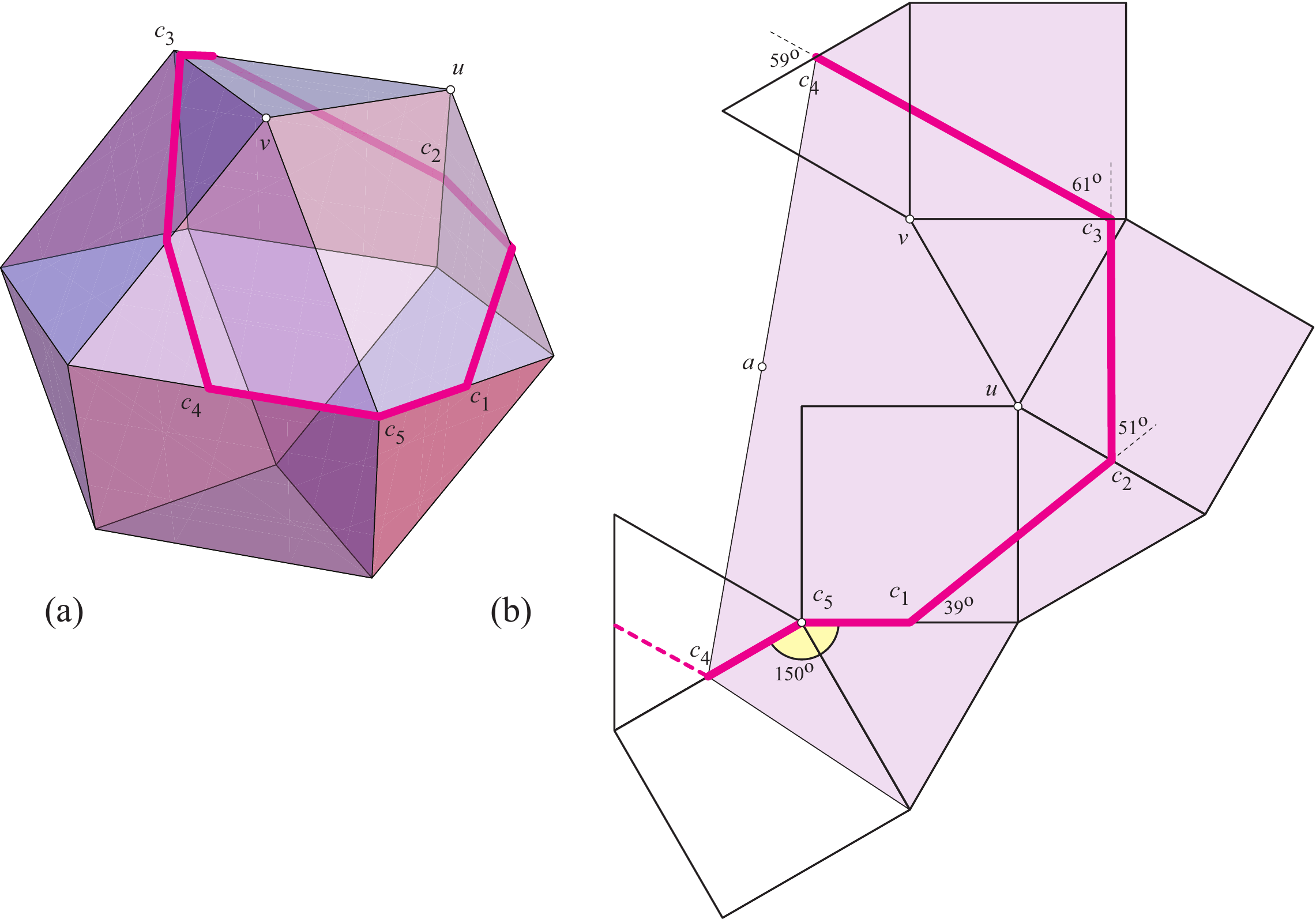}
\caption{(a)~A curve $C$ of five corners passing through one
polyhedron vertex.
$C$ is a convex to one side, and a reflex loop to the
other, with loop point $c_5$, at which $\b_5 = \frac{5}{6}\pi (=
150^\circ) < \pi$.
(b)~The cone $\L_R$ with apex $a$ is shaded.}
\figlab{CubeOctCurve2}
\end{figure}

\section{Summary and Extensions}

\subsection{Summarizing Theorem}
\noindent
Putting Theorems~\thmref{Convex} and~\thmref{Reflex} together,
we obtain:
\noindent

\begin{theorem}
For the following classes of curves $C$ on a convex polyhedron $\P$,
we may conclude that $C$ lives on a unique cone to both sides,
and is visible from the apex of each cone:
\begin{enumerate}
\squeezelist
\item $C$ is a quasigeodesic (because they are convex to both sides).
\item $C$ is convex and passes through no vertices 
(because then
the other side is reflex).
\item $C$ is convex and passes through one vertex
(because then the other side is a reflex loop whose other side is convex).
\item $C$ is convex and passes through several vertices
such that, at all but at most one corner $c_i$ of $C$,
$\a_i + \o_i \le \pi$.
In this situation, $C$ is a reflex loop to the other side because $\b_i \ge \pi$
at all but at most one vertex.
\end{enumerate}
\thmlab{ConvexReflex}
\end{theorem}

\subsection{Quasigeodesic Loops}
Our extension of the source unfolding of a
polyhedron~\cite{iov-sucpr-09}
(Section~\secref{Source} below)  holds for
classes of curves living on a cone to both sides,
while
our extension of the star unfolding of a polyhedron~\cite{iov-sucpql-10}
works for any quasigeodesic loop.
It is therefore natural to explore extending Theorem~\thmref{ConvexReflex} 
to encompass quasigeodesic loops.
Recall that
quasigeodesic loops are convex to one side, and
convex loops to the other.
Despite quasigeodesic loops being very special convex loops,
we show by example that there are quasigeodesic loops which
fail to satisfy Theorem~\thmref{ConvexReflex} in that they
do not live to a cone to both sides.

The construction is a modification of the example in 
Figure~\figref{ConvexLoop} showing that a convex loop might
not live on a cone.
In that example, $C$ is a convex loop to the left; we modify
the example so that it becomes convex to its right.
Let $\P$ be the polyhedron in Figure~\figref{ConvexLoop}(a).
Essentially we will retain $P_L$, the left half of $\P$, and replace
$P_R$ with a different surface to produce a new polyhedron $\P^*$.
Toward that end, add a new vertex $e$ at the midpoint of edge $ad$ of
$\P$.
Although we could make $e$ a true vertex with non-zero
curvature,
it is easiest to see the construction when $\o(e)=0$.
Let $C^*$ be the new curve, $C^*=(a,b,b',x,c',c,d,e)$, geometrically
the same as $C$ but now including $e$ on the path between $a$ and $d$.
So $C^*$ is still a convex loop to its left.
Let $\b = \angle b'xc'$ be the convex angle at the loop point $x$.

Now construct a planar convex polygon
$Q = (
\overline{a},
\overline{b},
\overline{b'},
\overline{x},
\overline{c'},
\overline{c},
\overline{d},
\overline{e}
)$,
each of whose edges has the same length as the corresponding 
edge of $C^*$---$|\overline{a}\overline{b}| = |ab|$, etc.---and such
that
$\angle \overline{b'}\overline{x}\overline{c'} = \b$, matching
$\angle b'xc'$.
These conditions do not uniquely determine $Q$, but any $Q$ that is
convex
and has angle $\b$ at $x$ suffices for the construction.
See Figure~\figref{QGCex3D}(a).
\begin{figure}[htbp]
\centering
\includegraphics[width=\linewidth]{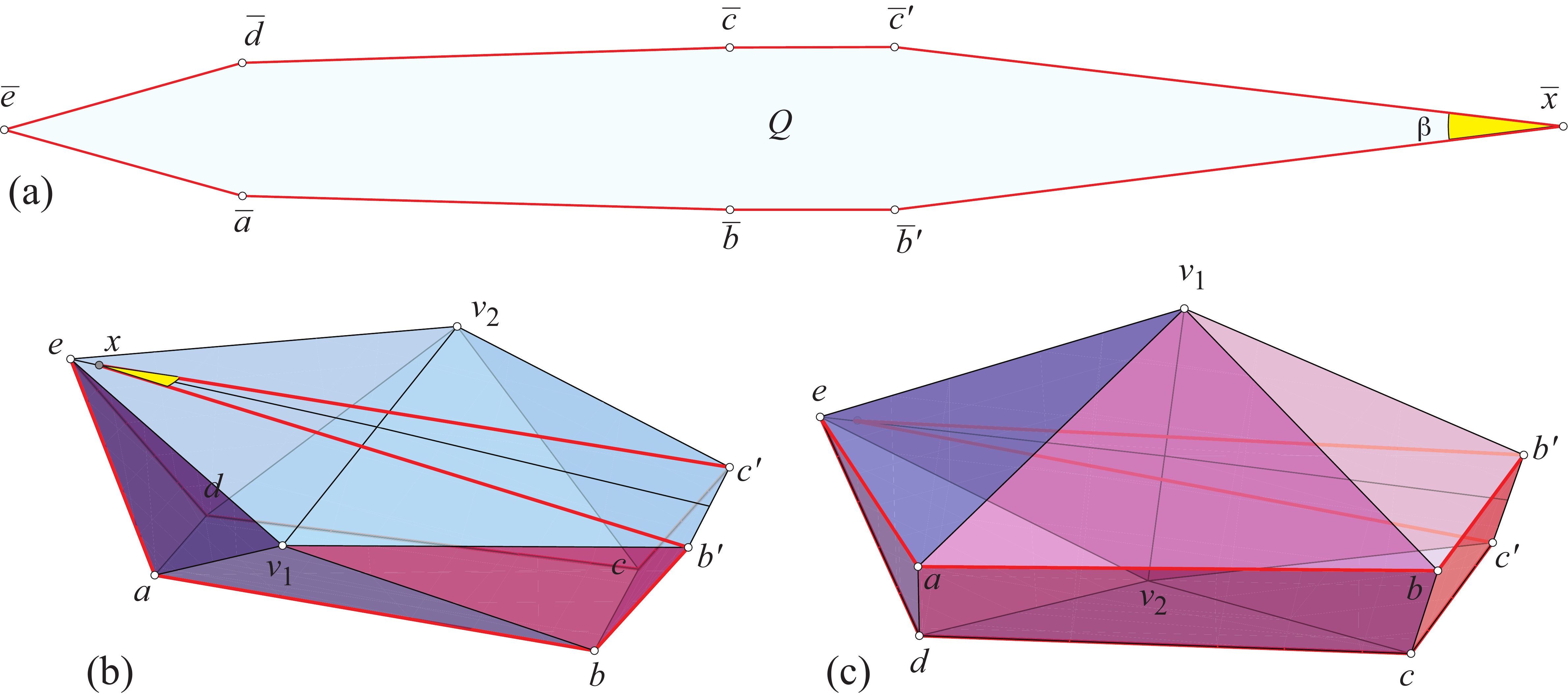}
\caption{
(a)~Convex polygon $Q$.
(b,c)~Two views of $\P^*$.
The dihedral angle at the ``roof edge'' $v_1 v_2$ was $\frac{1}{4}\pi$ in 
Figure~\protect\figref{ConvexLoop}(a) but is nearly $\pi$ in $\P^*$.
(The 3D shape here is only approximate, constructed via ad hoc computations.)
}
\figlab{QGCex3D}
\end{figure}

$\P^*$ is now constructed by gluing $P_L$, the top half of $\P$, to
$Q$,
matching corresponding vertices, $\overline{a}$ to $a$, etc.
Alexandrov's Gluing Theorem guarantees that the resulting surface
corresponds to a unique convex polyhedron $\P^*$.
Figure~\figref{QGCex3D}(b,c) shows an approximation to $\P^*$.
$C^*$ is a quasigeodesic loop on $\P^*$: a convex loop to the left
and convex by construction to the right.
$C^*$ lives on a (planar) cone to the right, but does not live on a
cone to its left for the same reason that $C$ did not on $\P$:
it does not fit.

We have established that convex loops always live on
the union of two cones,\footnote{
   Very roughly, we cut from the exceptional loop point
   $x$ via a geodesic to a point $y$ on $C$, yielding
    two convex curves $C_1$ and $C_2$ sharing $xy$,
   each of which lives on a cone.
   (This technique was used in~\cite{iov-sucpql-10}.)
} 
but we leave that a claim not pursued here.

\section{Applications}
\seclab{Applications}

\subsection{Development of Curve on Cone}
Nonoverlapping development of curves plays a role in unfolding
polyhedra without overlap~\cite{do-gfalop-07}.
Any result on simple (non-self-intersecting) development of curves
may help establishing nonoverlapping surface unfoldings.
One of the earliest results in this regard is~\cite{os-odcc3p-89},
which proved
that the left development
of a directed, closed convex curve 
does not self-intersect.
The proof used Cauchy's Arm Lemma.
The new viewpoint in our current work reproves this result without
invoking Cauchy's lemma, and extends it to a wider class of curves.

Every simple, closed curve $C$ drawn on a cone $\L$ and which encloses
the apex $a$ of $\L$ may be 
developed on the plane by rolling $\L$ on that plane.
More specifically,
select a point $x \in C$ and develop $C$ from $x$ back to $x$ again.
We call this curve in the plane $\overline{C}_x$.
Once $x$ is selected, the development is unique up to congruence in
the plane.
There is no distinction between right and left developments of a curve
on a cone;
that distinction only applies when there is nonzero curvature along $C$, as
there may be on the surface of a polyhedron $\P$.
If $g$ is a generator of $\L$ that meets $C$ in one point $\{x\}= g \cap
C$~---a
condition guaranteed by our visibility lemmas
(Lemmas~\lemref{ConvexVisible} and~\lemref{ReflexVisible})~---then
$\overline{C}_x$ is non-self-intersecting, because the unrolling of
the entire cone is non-overlapping.
Thus we obtain from Theorem~\thmref{ConvexReflex} a broader
class of curves on $\P$ that develop without intersection,
including reflex loops whose other side is convex.

\subsection{Overlapping Developments}
In general, 
$\overline{C}_x$ is not congruent to $\overline{C}_y$ when $x \neq y$.
We are especially interested in those $C$ for which
$\overline{C}_x$
is simple (non-self-intersecting) for every choice of $x$,
and we have just identified a class for which this holds.
Here we show that there exist $C$ such that $\overline{C}_x$ is
nonsimple for every choice of $x$.
We provide one specific example, but it can be generalized.

The cone $\L$ has apex angle $\a=\frac{3}{4}\pi$;
it is shown cut open and flattened in two views in
Figure~\figref{UnDevelopable}(a,b).
An open curve $C'=(p_1,p_2,p_3,p_4,p_5)$ is drawn on the cone.
Directing $C'$ in that order, it turns left by $\frac{3}{4}\pi$ at $p_2$,
$p_3$, and $p_4$.
\begin{figure}[htbp]
\centering
\includegraphics[width=\linewidth]{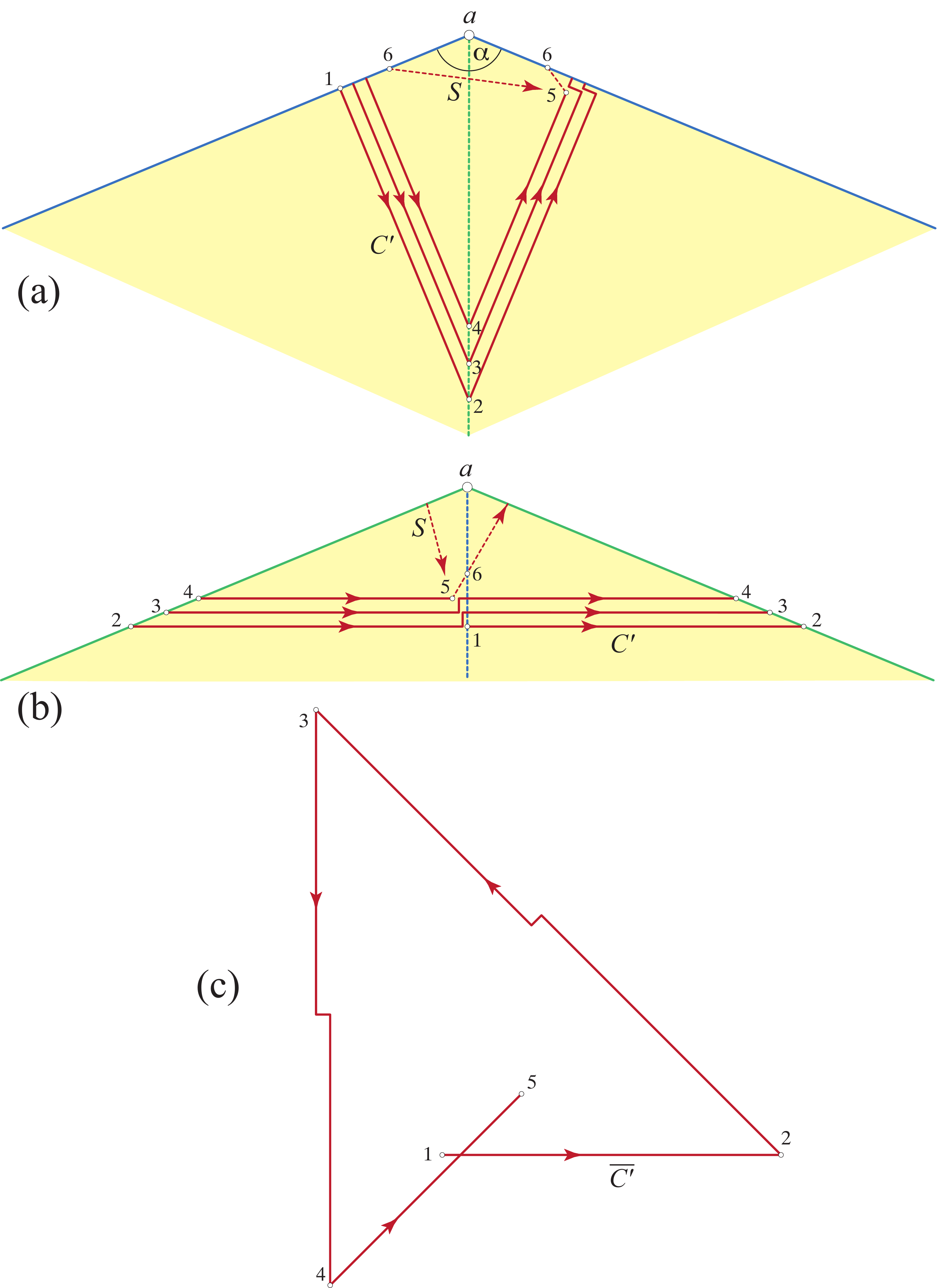}
\caption{(a)~Open curve $C'=(p_1,p_2,p_3,p_4,p_5)$ 
on cone of angle $\a$, with cone opened.
(b)~A different opening of the same cone and curve.
(c)~Development of curve $\overline{C'}$ self-intersects.}
\figlab{UnDevelopable}
\end{figure}
From $p_5$, we loop around the apex $a$ with a segment
$S=(p_5,p_6,p'_5)$,
where $p'_5$ is a point near $p_5$ (not shown in the figure).
Finally, 
we form a simple closed curve on $\L$ by then doubling $C'$
at a slight separation (again not illustrated in the figure), so that
from $p_5$ it returns in reverse order along that slightly displaced path to $p_1$
again.
Note that $C= C \cup S \cup C'$ is both closed and includes the apex
$a$
in its (left) interior.

Now, let $x$ be any point on $C$ from which we will start the
development $\overline{C}_x$.
Because $C$ is essentially $C' \cup C'$, $x$ must fall in one or the
other
copy of $C'$, or at their join at $p_1$.
Regardless of the location of $x$, at least one of the two copies of
$C'$ is unaffected.  So $\overline{C}_x$ must include $\overline{C'}$
as a subpath in the plane.

Finally, developing $C'$ reveals that it self-intersects: 
Figure~\figref{UnDevelopable}(c).
Therefore, $\overline{C}_x$ is not simple for any $x$.
Moreover, it is easy to extend this example to force self-intersection
for
many values of $\a$ and analogous curves.
The curve $C'$ was selected only because its development is self-evident.

\subsection{Source Unfolding}
\seclab{Source}
Every point $x$ on the surface of a convex polyhedron $\P$ leads
to a nonoverlapping unfolding called the \emph{source unfolding of
  $\P$ with respect to $x$}, obtained by cutting $\P$ along the cut
locus of $x$.
We can think of this as the \emph{source unfolding with respect to a
  point $x$.}
We have generalized in~\cite{iov-sucpr-09} this unfolding to unfold
$\P$ by 
cutting~---roughly speaking---~along
the cut locus of a simple closed curve $C$ on $\P$.
This unfolding is guaranteed to avoid overlap when $C$ lives on a cone
to both sides.  So it applies in exactly the conditions specified
in Theorem~\thmref{ConvexReflex}, and this is a central motivation
for our work here.

\section{Open Problems}
\seclab{Open}
We have not completely classified the curves $C$ on a convex
polyhedron $\P$ that live on a cone to both sides.
Theorem~\thmref{ConvexReflex} summarizes our results, but they are
not comprehensive.

\subsection{Slice Curves}
One particular class we could not settle are the slice curves.
A \emph{slice curve} $C$ is the intersection of $\P$ with a plane.
Slice curves in general are not convex.
The intersection of $\P$
with a plane is a convex polygon in that plane, but the surface angles
of $\P$
to either side along $C$ could be greater or smaller than $\pi$
at different points.
Slice curves were proved to develop without intersection,
to either side, in~\cite{o-dipp-03},
so they are strong candidates to live on cones.
However, we have not been able to prove that they do.
We can, however, prove that every convex curve on
$\P$ is a slice curve on some $\P'$
(this follows from \cite[Thm.~2, p.~231]{a-cp-05}),
and either side of any slice curve on $\P$ is the other
side of a convex curve on some $\P'$.

\subsection{Curve with a Nested Convex Curve}
We can extend the class of curves to which 
Lemma~\lemref{ConvexCone} (the convex-curve lemma)
applies beyond convex, but the extension is not truly substantive.
Let $C$ be a simple closed curve which encloses a convex curve $C'$
such that the region of $\P$ bounded between $C$ and $C'$ contains no
vertices.
See, e.g., Figure~\figref{IcosaZigZag}.
Then the proof of Lemma~\lemref{ConvexCone}
applies to $C'$ and $C$ lives on the same cone as $C'$.
\begin{figure}[htbp]
\centering
\includegraphics[width=0.45\linewidth]{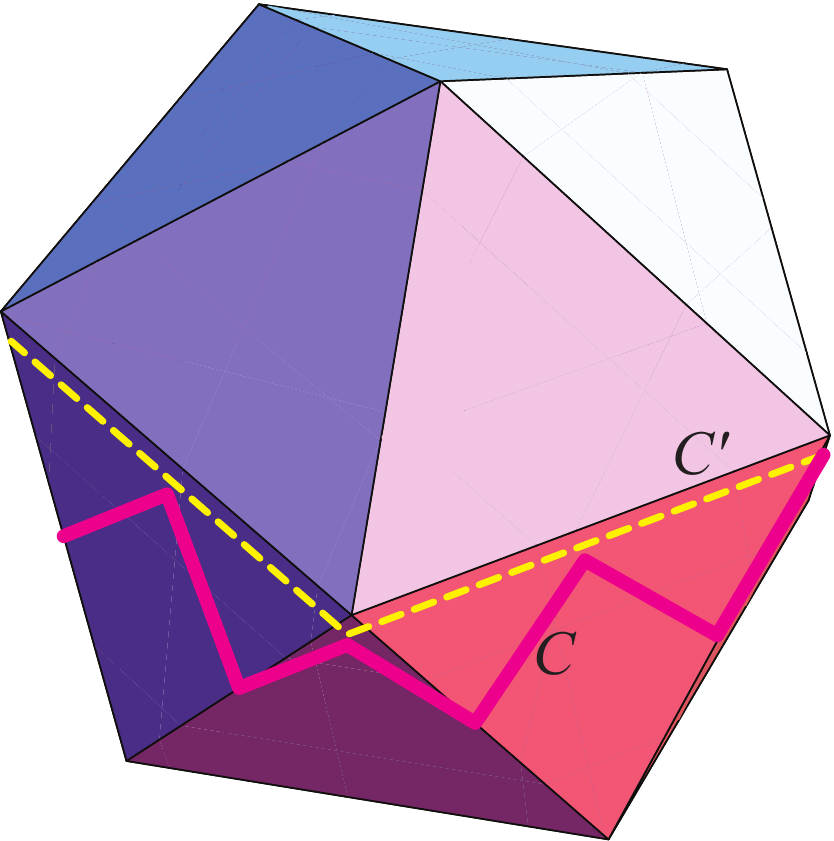}
\caption{$C'$ is convex (it is a geodesic) and $C$ lives on the
same cone (in this case a cylinder) as does $C'$.
}
\figlab{IcosaZigZag}
\end{figure}

\subsection{Cone Curves}
We have not obtained a complete classification of the curves on a cone that develop,
for every cut point $x$, as simple curves in the plane.
It would also be interesting to
identify the class of curves on cones for which there exists at
least one cut-point $x$
that leads to simple development.

\paragraph{Acknowledgments.}
We thank Jin-ichi Itoh, our coauthor 
on~\cite{iov-sucpr-09}
and~\cite{iov-sucpql-10}, the papers that motivated this investigation.

\bibliographystyle{alpha}
\bibliography{/Users/orourke/bib/geom/geom}
\end{document}